\RequirePackage{fix-cm}
\documentclass[smallextended]{svjour3}       
\smartqed  
\usepackage{amsmath}
\usepackage{amssymb}
\usepackage{graphicx}
\usepackage{tcolorbox}
\usepackage{tikz}
\usepackage{bbm}
\usetikzlibrary{arrows,calc,shapes, intersections}
\usetikzlibrary{decorations.pathreplacing}
\usepackage{etoolbox}
\usepackage{graphicx}
\usepackage{algorithm}
\usepackage{algpseudocode}
\usepackage{algorithmicx}
\usepackage{thmtools}

\DeclareMathOperator*{\argmax}{arg\,max}
 
\usepackage{cleveref}
\begin{document}

\title{Dominant Strategy Truthful, Deterministic Multi-armed Bandit Mechanisms with Logarithmic Regret for Sponsored Search Auctions}

\titlerunning{$\Delta$-UCB Multi-armed Bandit Mechanism}        

\author{Divya Padmanabhan$^{\dagger}$\and
Satyanath Bhat$^{*}$ \and 
      Prabuchandran K. J.$^{\ddagger}$ \and Shirish Shevade$^{\ddagger}$
      \and Y. Narahari$^{\ddagger}$
}


\institute{$^\dagger$ Singapore University of Technology and Design, Singapore\\
   $^*$ National University of Singapore \\$^{\ddagger}$ Indian Institute of Science (IISc), Bangalore
}

\date{Submitted: December 2018, Revised: May 2020}

\maketitle

\begin{abstract}
Stochastic multi-armed bandit (MAB) mechanisms are widely used in sponsored search auctions, crowdsourcing, online procurement, etc.  Existing stochastic MAB mechanisms with a deterministic payment rule, proposed in the literature, necessarily suffer a regret of $\Omega(T^{2/3})$, where $T$ is the number of time steps. This happens because the existing mechanisms consider the worst case scenario where the means of the agents' stochastic rewards are separated by a very small amount that depends on $T$. 
We make, and, exploit the crucial observation that in most scenarios, the separation between the agents' rewards is rarely a function of $T$. Moreover, in the case that the rewards of the arms are arbitrarily close, the regret contributed by such sub-optimal arms is minimal. 
Our idea is to allow the center to indicate the resolution, $\Delta$, with which the agents must be distinguished. This immediately leads us to introduce the notion of $\Delta$-Regret. Using sponsored search auctions as a concrete example (the same idea applies for other applications as well), we propose a dominant strategy incentive compatible (DSIC) and individually rational (IR), deterministic MAB mechanism, based on ideas from the Upper Confidence Bound (UCB) family of MAB algorithms.  Remarkably, the proposed mechanism $\Delta$-UCB achieves a $\Delta$-regret of $O(\log T)$  for the case of sponsored search auctions.  We first establish the results for  single slot sponsored search auctions and then non-trivially extend the  results to the case where multiple slots are to be allocated.
\end{abstract}

\section{Introduction}
\label{Introduction}
Multi-armed bandit (MAB) algorithms \cite{Bubeck2012a} are now widely used to model and solve problems where decisions are required to be made sequentially at every time step and there is an \textit{exploration - exploitation} dilemma. This   dilemma is the tradeoff that the planner faces in deciding whether to  explore arms that may yield higher rewards in the future or exploit the arms that have  already yielded high rewards in the past. If the rewards are generated from fixed distributions with unknown parameters, the setting goes by the name stochastic MAB \cite{Bubeck2012a}.   Popular algorithms in the stochastic MAB setting include Upper Confidence Bound (UCB) based algorithms \cite{UCBAuer2002} and Thompson Sampling \cite{jmlrShipraThompson} based algorithms. These algorithms incur $O(\log T)$ regret where $T$ is the total number of time steps. MAB algorithms are well studied with  several variants  \cite{Kleinberg2010,Bubeck2012b,Kapoor2018,chen2013combinatorial} and applications \cite{Santiago2017,PadmanabhanIJCNN16,scott2010modern,dirkx2018optimizing}. 

When the arms are controlled  by {\em strategic\/} agents, we need to tackle additional challenges. 
Mechanism design \cite{NARAHARI2014,NISAN2007,Nisan2007jair} has been applied in this context, leading to stochastic MAB mechanisms \cite{Liu2017}.
The design of such mechanisms requires ideas from online learning as well as mechanism design, both of which are increasingly gaining importance in the field of artificial intelligence.
An immediate application of stochastic MAB mechanisms is in sponsored search auctions (SSA). In SSA,  there are several advertisers who wish to display their ads along with the search results generated in response to a query from an internet user. In the standard model, an advertiser has only one ad to display. We use the terms agent, ad, and advertiser interchangeably. There are two components that are of interest to the planner or the search engine, (1) {\em stochastic component\/}: click through rate (CTR) of the ads or the probability that a displayed ad receives a click (2) {\em strategic component\/}: valuation of the agent for every click that the agent's ad receives. The search engine would seek to allocate a slot to an ad which has the maximum social welfare (product of click through rate and valuation). However neither the CTRs nor the valuations of the agents are known. This calls for a learning algorithm to learn the stochastic component (CTR) as well as a mechanism to elicit the strategic component (valuation). This problem could become much harder as the agents may manipulate the learning process \cite{Babaioff2014,Jain2016} to gain higher utilities.

For single slot SSA, it is known that any deterministic MAB mechanism (that is, a MAB mechanism with a deterministic allocation and payment rule) suffers a regret \cite{Bubeck2012a,Feldman2014} of  $\Omega(T^{2/3})$ \cite{Babaioff2014}. Furthermore, there exists a deterministic MAB mechanism with regret matching the theoretical lower bound \cite{Babaioff2014} and also satisfies ex-post truthfulness, the strongest notion of truthfulness (a posteriori to the clicks). When a more relaxed notion of truthfulness is targeted (truthfulness in expectation of the clicks), the regret guarantee improves to $O(T^{1/2})$~\cite{babaioff2010truthful}.  Truthfulness in expectation has also been achieved in \cite{Gonen2007,Pavlov2009}. The regret can be further
improved when randomized mechanisms are used and in fact the regret in this space is $O(\log T)$ \cite{babaioff2010truthful,Jain2018}.
However, the high variance that is inevitable to the payments in randomized
mechanisms is a serious deterrent to the use of randomized mechanisms. Towards reducing the variance, \cite{Ghalme2017} propose a MAB mechanism using Thompson sampling \cite{jmlrShipraThompson}. However the notion of truthfulness achieved is `within period DSIC' and with high probability. Thus again, only a weaker notion of truthfulness is achieved compared to ex-post truthfulness.

In this work, we observe that the
characterization provided by Babaioff et al. \cite{Babaioff2014} targets the worst case scenario. In particular, in the lower bound proof 
of regret of $\Omega(T^{2/3})$, they consider an example scenario where the actual 
separation, $\bar{\Delta}$, between the expected rewards of the arms  is a
function of $T$. We note that when a similar example ($\bar{\Delta} = T^{-1}$) is used with the popular UCB algorithm~\cite{UCBAuer2002}, the number of pulls of sub-optimal arms could be linear, even in the non-strategic case.
Hence, a dependence of $\bar{\Delta}$ on $T$ is severely restrictive for the case
when the rewards are stochastic, even when the arms are non-strategic. We make the
observation that $\bar{\Delta}$ is in most situations independent of
$T$. This motivates our main idea
in this paper, which is to
provide the planner an option to specify a parameter $\Delta$, which is the tolerance or
distinguishing level for sub-optimal arms. The understanding is that any arm that is within
$\Delta$ from the best arm will not cause any additional regret to the planner. For example, the best arm may yield  expected reward of $6.000$ while a sub-optimal arm may yield a very close expected reward of $5.999$. The planner is typically indifferent to such small differences. Traditional exploration-separated {\color{blue} schemes} end up spending a huge number of exploration rounds in order to distinguish between these two closely separated arms. 
\begin{tcolorbox}
\textbf{Setting the value of $\Delta$: An Example \\}
The value of $\Delta$ is set by the central planner depending on how well he would like to distinguish between the arms. For example, consider the case where there are two agents. Agent $1$ has a CTR $\mu_1 = 0.8$ and valuation for every click $\theta_1 = 5$ units. Agent $2$ has a CTR $\mu_2 = 0.3999$ and a valuation for every click $\theta_2 = 10$. Agent 1 is the more preferred agent as his expected social welfare is $\mu_1 \theta_1 = 4$ while the expected social welfare for agent 2 is $\mu_2 \theta_2 = 3.999$. 
 Then the actual separation between the agents, $\bar{\Delta} = 4 - 3.999 = 0.001$. But the planner may be indifferent to such a small difference of $0.001$ in expected social welfare. Therefore he would be satisfied with selecting either of the agents. Hence, he should set the parameter $\Delta$ to any value greater than $0.001$.
 \end{tcolorbox}
 This notion of $\Delta$ tolerance will require an appropriate definition of regret, which we call
$\Delta$-regret. Focussing on $\Delta$-regret instead of   the usual notion of regret helps us to reduce the number of exploration rounds significantly from $O(T^{2/3})$ to $O(\log T)$.
We propose an exploration separated mechanism based on UCB, which achieves a $\Delta$-regret
of $O(\log T)$.  This mechanism can be readily applied in several settings such as SSA,
crowdsourcing, and online procurement. For the rest of the paper, however,
we use SSA as a running example.

\subsection*{Contributions:}
(1) 
We make the crucial observation that in most MAB scenarios, the separation between the agents' rewards is rarely a function of $T$ (the number of time steps). Moreover, in the case that the rewards of the arms are arbitrarily close, the regret contributed by such sub-optimal arms is negligible. We exploit this observation
to allow the center to specify the resolution, $\Delta$, with which the agents must be distinguished. We introduce the notion of $\Delta$-Regret to formalize this regret. \\
(2) Using sponsored search auctions as a concrete example, we propose a dominant strategy incentive compatible (DSIC) and individually rational (IR) MAB mechanism with a deterministic allocation and payment rule, based on ideas from the UCB family of MAB algorithms.  The proposed mechanism $\Delta$-UCB achieves a $\Delta$-regret of $O(\log T)$  for the case of single slot sponsored search auctions. The truthfulness achieved by $\Delta$-UCB is a posteriori to the click realizations and is the strongest form of truthfulness. {\color{blue} This loss of $O(\log T)$ would not have been possible otherwise if the traditional notions of regret were used.  In particular the number of exploration rounds in $\Delta$-UCB is $O(\log T)$ as opposed to the $O(T^{2/3})$ rounds which were mandatory so far for ensuring a truthful, deterministic mechanism.  Thus we now enable the planner to be relieved from this huge number of exploration rounds. We also show that a lower bound on the $\Delta$-regret suffered by any mechanism is  $\Omega(\log T)$}.\\ 
(3) We  non-trivially extend the  above results to the case where multiple slots are to be allocated.
 Here again, our mechanism is DSIC, IR, and achieves  a $\Delta$-regret that is $O(\log T)$. 

Our results are generic to stochastic MAB  mechanisms and can be applied to other popular applications such as crowdsourcing and online procurement.

\section{Relevant Work}
In the area of MAB mechanisms, a lot of work has been done in sponsored
search auctions. Babaioff et. al.\cite{Babaioff2014} provide a
characterization of truthful MAB mechanisms, wherein the objective is
to maximize social welfare. They introduce the notion of influential
rounds. The influential rounds are the rounds where the  parameters of
reward distributions (CTRs) are learnt. One of the characterizations of
truthful deterministic mechanisms is that the allocation must be
exploration separated, that is, in such  influential rounds, the
allocation must not depend on the bids of the agents. The allocation
is also required to be point wise monotone. One of the main results of
their paper is that any truthful, deterministic MAB mechanism incurs a
 regret of $\Omega(T^{2/3})$.  {\color{blue} In particular, their analysis holds an adversarial nature,
 as the sub-optimality between the best and second best arm is chosen as if by an adversary,  to be proportional to $T^{-1/3}$. Such a choice 
ensures a huge regret for any truthful, deterministic mechanism. }
They also provide a mechanism which incurs a matching upper bound
regret  of $O(T^{2/3})$.  Devanur et. al. \cite{Devanur2009} concurrently
provide similar bounds on the regret when the objective is revenue
maximization rather than social welfare maximization.

All the above results pertain to the setting of single slot
auctions where there is a single slot for which
the agents compete. In the generalization of this setting 
multiple slots are reserved for ads. This setting is more challenging as every slot is not identical and some slots are more prominent than the others. MAB mechanisms have also
been extended to the multiple slot setting \cite{Gatti2012} in line
with the characterization in \cite{Babaioff2014}. Hence, a similar
regret  of $O(T^{2/3})$ on the social welfare has been attained here
as well. Similar results are also stated in the characterisation provided in  \cite{Akash2012}.

MAB mechanisms have also been proposed in the context of crowdsourcing
\cite{Biswas2015}. Some of these mechanisms incur a regret of $O(\log T)$. This is rendered possible due to the specific nature of the problem in hand. In particular,  Bhat et. al. \cite{BhatAAMAS16} look at divisible tasks. 
Jain et. al. \cite{JainAAMAS16} look at deterministic mechanisms where a block of
tasks is allocated to each agent and provide a weaker notion of
truthfulness.  

The lower bound of both of social welfare regret as well as regret in
the revenue of $\Omega(T^{2/3})$ have influenced subsequent
research to follow similar assumptions and thereby obtain a similar regret. However, we show in this work that it is indeed possible to design a deterministic mechanism which attains logarithmic regret and is also truthful in the dominant strategy incentive compatible (DSIC) sense \cite{Myerson1991}. DSIC, of course, is the most preferred form of truthfulness \cite{NARAHARI2014}. This work opens up the possibility for a planner to move away from the worst case scenario to a more realistic scenario. We enable the  planner to specify a resolution parameter for distinguishing the arms, introduce the notion of $\Delta$-regret and thereafter propose a mechanism that ensures that the number of exploration rounds and hence the regret suffered is only $O(\log T)$ instead of the  expensive $\Omega(T^{2/3})$ available currently in state of the art. {\color{blue} We summarize the contrast between our work and the state of the art in \Cref{tab:comparison}}.
\begin{table}[h!]
{
\color{blue}
\begin{tabular}{|p{0.18\linewidth}|p{0.35\linewidth}|p{0.35\linewidth}|}\hline
& \cite{Babaioff2014} & Our work \\\hline\hline
Loss studied & Regret & $\Delta$-regret \\\hline
Additional parameters & None & $\Delta$: tolerance specified by the planner \\\hline
Mechanism properties & DSIC, deterministic, exploration separated,  $O(T^{2/3})$ exploration rounds & DSIC, deterministic, exploration separated,  $O(\log T)$ exploration rounds \\\hline
Upper bound on loss  &  $O(T^{2/3})$ & $O(\log T)$ \\\hline
Lower bound on loss &  $\Omega(T^{2/3})$ & $\Omega(\log T)$ \\\hline
\end{tabular}
\caption{Comparison of our results with state of the art}
\label{tab:comparison}
}
\end{table}

\section{The Model: Single Slot SSA}
We now describe our SSA setting.  For ease of reference, our notations are provided in \Cref{tab:notations-single-slot}. Let  $K$ be the number of agents or arms. We denote the set of arms by $[K]$. Each of the $K$ arms, when pulled, gives rewards from distributions with unknown parameters. We assume here, that the form of the distributions are known but the parameters of the distributions are unknown. In SSA, the rewards of the arms correspond to clicks. The clicks for the advertisements are assumed to be generated from Bernoulli distributions with parameters $\mu_1, \mu_2, \ldots, \mu_K$ where $\mu_i$ is the CTR or probability that advertisement $i$ receives a click once observed. The means $\mu_1, \ldots, \mu_K$ are unknown.

A click realization $\rho$ represents the click information of every agent at all rounds, that is, $\rho_i(t) =1 $ if agent $i$ received a click in round $t$. In a round $t$, only the click information of the allocated agent is revealed after the completion of the round. Click information of all other unallocated agents is never known to the planner.

The agents also have their valuations for each click they receive. We work in the `pay per click' setting where the  agent pays the search engine for each click received. Let the true valuation of agent $i$ be $\theta_i$ for a click. $\theta_i$ is a private type of agent $i$ and is never known to the learner. However the agent is asked to bid his valuation. Let the bid of agent $i$ be $b_i$. We denote by a vector $b = (b_1, \ldots, b_K)$ the bid profile of all the agents.  The central planner wants to ensure that the agents bid their true valuations, that is $b_i$ must be equal to $\theta_i$. Assume that there is a single slot which must be allocated to one of the $K$ agents. We denote by $W_i$ the social welfare when agent $i$ is allocated a slot, that is, $W_i = \mu_i \theta_i$. The social welfare represents the expected valuation of agent $i$ per click. If the CTRs of the agents as well as their valuations were known, the planner would have selected the arm with the maximum social welfare, that is, $\mu_i \theta_i$. However neither $\mu_i$ nor $\theta_i$ is known to the planner. Assume $\theta_{max}$ is the maximum valuation that any agent can have and is common knowledge. The central agent wants to allocate a single slot to one of the ads in such a way that the net social welfare of the allocation is maximized.

A mechanism $\mathcal{M} = \left\langle\mathcal{A}, P\right\rangle$ is a tuple containing an allocation rule  $\mathcal{A}$ and a payment rule $P$. At every time step or round $t$, the allocation rule acts on a bid profile $b$ of the agents as well as click realization $\rho$ and allocates the slot to one of the $K$ agents, say $i$. Then $\mathcal{A}(b, \rho, t) = i$. Alternatively we denote the indicator variable $\mathcal{A}_i(b, \rho, t) = \mathbbm{1} [\mathcal{A}(b, \rho, t) = i]$.  The payment rule $P^t = (P_1^t, P_2^t, \ldots, P_K^t)$, where $P_i^t(b, \rho)$ is the payment to be made by agent $i$ at time $t$ upon receiving a click, when the bids are $b$ and for click realization $\rho$. As stated earlier $\rho_i(t)$ of the allocated agent alone is observed. Also note that the allocation as well as payments in each round $t$ only depends on the click histories till that round.

\begin{table}[h!]
\begin{tabular}[c]{|p{0.1\textwidth}|p{0.8\textwidth}|}
\hline
\textbf{Symbol} & \textbf{Description} \\\hline
$K$, $[K]$ & No. of agents and agent set \\\hline
$\mu_i$ & CTR of agent $i$ \\\hline
$\theta_{i}$  & Valuation of agent $i$ for each click \\\hline
$W_i$ & Social welfare when agent $i$ is allocated \\
\hline
$\rho_i(t)$ & Click realization of agent $i$ at time $t$ \\\hline
$\theta_{max}$ & Maximum valuation over all agents = $\max_i \theta_i$\\\hline
$b_i$ & Bid of agent $i$ \\\hline
$b$ & Bid profile of all agents \\\hline
$b_{-i}$ & Bid profile of all agents except agent $i$ \\\hline 
$N_{i,t}$ & No. of times agent $i$ has been selected till time $t$ \\\hline 
$\mathcal{A}(b, \rho, t)$ & Allocation at time $t$ for bid profile $b$ and click realization $\rho$ \\\hline
$i_*$ & Agent with maximum social welfare. Ideally $i_*$  must be allocated at every time step \\\hline
$W_*$ & Social welfare when agent $i_*$ is allocated \\\hline
$\Delta$ & Input parameter by center to indicate the level at which the agents must be distinguished \\\hline
$S_\Delta$ & Set of agents whose social welfare is less than $\Delta$ away from $i_*$. These agents do not contribute to $\Delta$-regret. \\\hline
$\widehat{\mu}_{i,t}^+ $ & UCB index corresponding to $\mu_i$ at time $t$\\\hline
$\widehat{\mu}_{i,t}^- $ & LCB index corresponding to $\mu_i$ at time $t$\\\hline
$\widehat{\mu}_{i,t}$ & Empirical CTR of agent $i$ estimated from samples up to time $t$  \\\hline 
$P_i^t$ & Payment charged to agent $i$ if he is allocated a slot at time $t$ and he gets a click \\\hline
\end{tabular}
\caption{Notations for the single slot SSA setting}
\label{tab:notations-single-slot}
\end{table}

Let $i_*$ be the arm with the largest social welfare, that is, $i_* = \argmax\limits_{i \in [K]} W_i$. We denote the corresponding social welfare as $W_* = \max_{i \in [K]} W_i $. We denote by $I_t$ the agent chosen at time $t$ as a shorthand for $\mathcal{A}(b, \rho, t)$. For any given $\Delta > 0$, define the set $S_\Delta = \{i \in [K]: W_* - W_i < \Delta \}$. $S_\Delta$ denotes the set of all agents separated from the best arm $i_*$ with a social welfare less than $\Delta$. These arms are therefore indistinguishable for the center and they contribute zero to the regret. Note that $\Delta$ is a parameter that the center fixes based on the amount in dollars he is willing to tradeoff for choosing sub-optimal arms, given he has only a fixed time horizon $T$ to his disposal. To capture this revised and more practical notion of regret, we introduce the metric $\Delta$-regret. Formally,
\begin{align}
\Delta\text{-regret} = \sum_{t=1}^T (W_* - W_{I_t}) \mathbbm{1}\left[I_t \in [K]\setminus S_\Delta\right]
\end{align}

The center may not want to invest a huge number of exploration rounds ($\Omega(T^{2/3})$ in state of the art) to perfectly distinguish the arms that are arbitrarily close. Many a time, the planner may instead be willing to allocate arms that are at most $\Delta$ away from the best arm.  The center therefore suffers a regret only when an agent with a social welfare greater than $\Delta$ away from $W_*$ is chosen. $\Delta$-regret captures this loss. 

The goal of our mechanism is to select agents at every round $t$ to minimize the $\Delta$-regret. 
\section{Our Mechanism: $\Delta$-UCB}
We are now ready to describe our mechanism $\Delta$-UCB. The idea in $\Delta$-UCB is to explore all the arms in a round-robin fashion for a fixed number of rounds. The number of exploration rounds is fixed based on the desired $\Delta$, specified by the planner. At the end of exploration, with high probability, we are guaranteed that  the arms not in  $S_\Delta$ are well separated from the best arm $i_*$ with respect to their social welfare estimates. In the exploration rounds, agents need not pay and these rounds are free. 

Further on, for all the remaining rounds, the best arm as per the UCB estimate of social welfare is  chosen. However in the exploitation rounds, the chosen agent pays an amount for each click he receives. The amount to be paid by the agent is fixed based on variant of the well known Vickrey Clark Grove (VCG) scheme \cite{vickrey1961counterspeculation} known as weighted VCG \cite{NISAN2007}. Note that no learning takes place in these rounds and the UCB, LCB indices do not change thereafter.
We present our mechanism in \Cref{alg:delta-ucb-single-slot}.
\begin{algorithm}[h!]
\begin{algorithmic}
\Require{ \\$T$: Time horizon, $K$: number of agents \\$\Delta:$ parameter fixed by the center \\ $\theta_{max}:$ Maximum valuation of the agents}
\hrule \\
\State Elicit bids $b= (b_1, b_2, \ldots, b_K)$ from all the agents
\State Initialize $\widehat{\mu}_{i,0}  = 0 , N_{i,0}  = 0 \; \forall i \in [K]$
\State $\gamma = \lceil 8K\theta_{max}^2\log T/\Delta^2 \rceil$
\For{$t= 1, \ldots, \gamma}$ \Comment Exploration rounds
\State $I_t = ((t-1) \mod K )+ 1 $ \Comment Round-robin exploration
\State $N_{I_t, t} = N_{I_t, t-1} + 1$
\State $\mathcal{A}(b, \rho, t) = I_t$ \Comment Allocate slot to agent $I_t$ and observe $\rho_{I_t}(t)$
\State $\widehat{\mu}_{I_t, t} = (\widehat{\mu}_{I_t,{t-1}}N_{I_t, t-1} + \rho_{I_t}(t))/N_{I_t, t}$
\State $\epsilon_{I_t, t} = \sqrt{2\log T/N_{I_t,t}}$
\State $\widehat{\mu}_{I_t,t}^+ = \widehat{\mu}_{I_t,t} + \epsilon_{I_t, t} $
\State $\widehat{\mu}_{I_t,t}^- = \widehat{\mu}_{I_t,t} - \epsilon_{I_t, t} $
\State $\widehat{\mu}_{i,t}^+ = \widehat{\mu}_{i,t-1}^+  \;\forall i \in [K]\setminus \{I_t\} $
\State $\widehat{\mu}_{i,t}^- = \widehat{\mu}_{i,t-1}^- \;\forall i \in [K]\setminus \{I_t\}$
\State $P_{i}^t (b, \rho) = 0 \;\forall i \in [K]$
 \Comment Free rounds
\EndFor
\State $\hat{i}_* = \argmax_{i \in [K]} \widehat{\mu}_{i,\gamma}^+ b_i$
\State $j= \argmax_{i \in  [K]\setminus \{\widehat{i}_*\}} \widehat{\mu}_{i,\gamma}^+ b_i $
\State $P =  \widehat{\mu}_{j,\gamma}^+ b_j /\widehat{\mu}_{\hat{i}_*,\gamma}^+$
\For{$t = \gamma+1, \ldots, T$} \Comment Exploitation rounds
\State $\mathcal{A}(b, \rho, t) = \hat{i}_*$
\State $P_{ \hat{i}_*}^t(b, \rho) = P \times \rho_{\hat{i}_*}(t)$
\Comment Agent pays only for a click
\State $P_{i}^t (b, \rho) = 0 \; \forall i \in [K] \setminus \{\hat{i}_*\}$
\State $\widehat{\mu}_{i,t}^+= \widehat{\mu}_{i,\gamma}^+$, $\widehat{\mu}_{i,t}^-= \widehat{\mu}_{i,\gamma}^- \; \forall i \in [K]$  \Comment No more learning
\EndFor
\end{algorithmic}
\caption{$\Delta$-UCB Mechanism for single slot SSA}
\label{alg:delta-ucb-single-slot}
\end{algorithm}

\subsection{Properties of $\Delta$-UCB}
\label{sec:truthfulness}
Next we discuss the properties satisfied by $\Delta$-UCB regarding truthfulness and regret. Before that, we state a few useful definitions which will help in understanding the notion of truthfulness. 

At any time step, every agent obtains some utility by participating in the mechanism. This utility is a function of his bid, valuation, bids of other agents and his click realization. Let $\Theta_i$ denote the space of bids of agent $i$. $b_{-i} = (b_1, \ldots, b_{i-1}, b_{i+1}, \ldots, b_K)$ is the bid profile containing bids of all agents except agent $i$. Let $\Theta_{-i}$ denote the space of bids of all agents other than agent $i$. Therefore 
$\Theta_{-i} = \Theta_1 \times \ldots, \times \Theta_{i-1} \times \Theta_{i+1}\times \ldots\times \Theta_K$.
We denote by $u_i(b_i, b_{-i}, \rho, t; \theta_i)$ the utility to agent $i$ at time $t$ when his bid is $b_i$, his valuation is $\theta_i$, the bid profile of the remaining agents is $b_{-i}$ and the click realization is $\rho$. 
All agents are assumed to be rational and are interested in maximizing their own utilities.

In our setting the utility to an agent $i$ is computed as,
\begin{align}
\label{eqn:utility}
u_i(b_i, b_{-i}, \rho, t; \theta_i) = (\theta_i - P_i^t(b, \rho)) 
\mathcal{A}_i(b_i, b_{-i}, \rho, t)\rho_i(t)
\end{align}
The idea behind the computation of the utility is as follows. If an agent $i$ does not receive an allocation (that is, $\mathcal{A}_i(b_i, b_{-i}, \rho, t)=0$), his utility is also zero. He gets a non-zero utility only if he receives an allocation. If he receives an allocation and also a click ($\rho_i(t) = 1$), then his utility is the difference between his valuation for the click and the amount he has to pay to the search engine ($\theta_i - P_i^t(b, \rho)$). If he does not receive a click ($\rho_i(t) = 0$), his utility is zero.
\begin{definition}{Dominant Strategy Incentive Compatible (DSIC) \cite{Babaioff2014}:}
A mechanism $M = \left\langle \mathcal{A},P \right\rangle$ is said to be dominant strategy incentive compatible if $\forall i \in [K], \forall b_i \in \Theta_i$, $\forall b_{-i} \in \Theta_{-i}, \forall \rho, \forall t, u_i(\theta_i, b_{-i}, \rho,t; \theta_i) \geq u_i(b_i, b_{-i}, \rho,t; \theta_i)$.
\end{definition}
Note that in the above definition, the truthfulness is demanded a posteriori to even the click realization \cite{Gatti2012}. Hence it is the strongest notion of truthfulness. Examples for weaker forms of truthfulness include those which take expectation over click realizations.

\begin{definition}{Individually Rational (IR):}
A mechanism $M = \left\langle\mathcal{A},P \right\rangle$ is said to be individually rational if $\forall i \in [K]$, $\forall b_{-i} \in \Theta_{-i},\forall \rho, \forall t, u_i(\theta_i, b_{-i}, \rho,t; \theta_i) \geq 0$.
\end{definition}

\begin{theorem}
\label{theorem:single-slot-truthful}
$\Delta$-UCB mechanism is dominant strategy incentive compatible (DSIC) and individually rational (IR).
\end{theorem}	
\begin{proof}
We analyze the scenarios where an agent $i$ bids his true valuation and receives an allocation and also when he does not. We show that in both these scenarios, bidding his true valuation $\theta_i$ is indeed a best response strategy.
We only need to consider the exploitation rounds because in the exploration rounds, every agent is allocated a fixed number of rounds independent of his bids and these rounds are also free for agents.\\
\underline{\textbf{Case 1:}} $\mathcal{A}_i(\theta_i, b_{-i}, \rho, t)=1$ \\
This implies that when the agent bids his true valuation, he gets an allocation. Therefore $\widehat{\mu}_{i,t}^+ \theta_i >
\widehat{\mu}_{l,t}^+ b_l$  for all the other agents $l$. In particular, let agent $j$ be such that $j = \argmax_{l \in [K] \setminus \{i\}} \widehat{\mu}_{l,t}^+ b_l$. The amount to be paid by agent $i$ is $P_i^t(\theta_i, b_{-i}, \rho) = \widehat{\mu}_{j,t}^+ b_j / \widehat{\mu}_{i,t}^+$. If he receives a click then 
 $u_i(\theta_i, b_{-i}, \rho,t; \theta_i)= \theta_i - \widehat{\mu}_{j,t}^+ b_j / \widehat{\mu}_{i,t}^+ > 0 $.\\ 
\underline{Overbid:} If agent $i$ bids a value $b_i> \theta_i$, he continues to receive an allocation and his payment is still the same, 
$P_i^t(b_i, b_{-i}, \rho) = \widehat{\mu}_{j,t}^+ b_j / \widehat{\mu}_{i,t}^+$. Therefore his utility continues to be  $u_i(b_i, b_{-i}, \rho,t; \theta_i)= \theta_i - \widehat{\mu}_{j,t}^+ b_j / \widehat{\mu}_{i,t}^+ = u_i(\theta_i, b_{-i}, \rho,t; \theta_i)$. Therefore he does not benefit from an overbid. \\
\underline{Underbid:} Suppose agent $i$ bids a value $b_i< \theta_i$. 

\emph{Case a:} If $b_i$ is such that $\widehat{\mu}_{i,t}^+b_i < \widehat{\mu}_{j,t}^+ b_j$, the he fails to get an allocation as $\mathcal{A}(b_i, b_{-i}, \rho, t) =j \neq i $. Then the utility to agent $i$ is $u_i(b_i, b_{-i}, \rho,t; \theta_i) =0 <
u_i(\theta_i, b_{-i}, \rho,t; \theta_i)$. Therefore he clearly loses his utility by such an underbid. 

\emph{Case b:} Suppose $b_i$ is such that $\widehat{\mu}_{i,t}^+\theta_i > \widehat{\mu}_{i,t}^+b_i > \widehat{\mu}_{j,t}^+ b_j$. That is agent $i$ bids in such a way that he wins the allocation even with an underbid. Then, if he gets a click, the amount he must pay to the center is $P_i^t(b_i, b_{-i}, \rho) = \widehat{\mu}_{j,t}^+ b_j / \widehat{\mu}_{i,t}^+ $. Therefore his utility $u_i(b_i, b_{-i}, \rho,t; \theta_i)= \theta_i - \widehat{\mu}_{j,t}^+ b_j / \widehat{\mu}_{i,t}^+= u_i(\theta_i, b_{-i}, \rho,t; \theta_i) $. He obtains the same utility as a truthful bid and there is no benefit from such an underbid.\\
\underline{\textbf{Case 2:}} $\mathcal{A}_i(\theta_i, b_{-i}, \rho, t)=0$ \\
This implies that when the agent bids his true valuation, he does not get an allocation. Suppose agent $j$ wins the allocation. $\mathcal{A}(\theta_i, b_{-i}, \rho, t)=j$ and $\widehat{\mu}_{i,t}^+ \theta_i <
\widehat{\mu}_{j,t}^+ b_j$. \\
\underline{Truthful bid:} Since agent $i$ does not win an allocation with a truthful bid, his utility $u_i(\theta_i, b_{-i}, \rho,t; \theta_i)=0$\\
\underline{Overbid:} Suppose agent $i$ bids in such a way that $b_i > \theta_i$. We have two sub-cases here.

\emph{Case a:} If $b_i$ is such that $\widehat{\mu}_{i,t}^+ \theta_i <
\widehat{\mu}_{j,t}^+ b_j < \widehat{\mu}_{i,t}^+ b_i $, then agent $i$ wins the allocation. So, $\mathcal{A}_i(b_i, b_{-i}, \rho, t)=1$. If he gets a click, he now has to make a payment $P_i^t(b_i, b_{-i}, \rho) = \widehat{\mu}_{j,t}^+ b_j/ \widehat{\mu}_{i,t}^+ $. Now his utility $u_i(b_i, b_{-i}, \rho,t; \theta_i) = \theta_i - \widehat{\mu}_{j,t}^+ b_j/ \widehat{\mu}_{i,t}^+ $ $< 0 $. And in particular $ u_i(b_i, b_{-i}, \rho,t; \theta_i)< u_i(\theta_i, b_{-i}, \rho,t; \theta_i) $ $= 0$. Therefore, such an overbid  is clearly disadvantageous compared to a truthful bid.

\emph{Case b:} Suppose $\widehat{\mu}_{i,t}^+ \theta_i < \widehat{\mu}_{i,t}^+ b_i < \widehat{\mu}_{j,t}^+ b_j $. The overbid by agent $i$ is not sufficient to make him win the allocation and agent $j$ wins the allocation, $\mathcal{A}(b_i, b_{-i}, \rho, t)=j$. The utility of agent $i$, $ u_i(b_i, b_{-i}, \rho,t; \theta_i)=0 = u_i(\theta_i, b_{-i}, \rho,t; \theta_i)$. Therefore there is no advantage for agent $i$ by this case of overbid.\\
\underline{Underbid:} If agent $i$ bids in such a way that $b_i < \theta_i$, he continues to lose the allocation and therefore his utility,$u_i(b_i, b_{-i}, \rho,t; \theta_i)=0
= u_i(\theta_i, b_{-i}, \rho,t; \theta_i)$. Since, the utility by an underbid remains the same as a truthful bid, there is clearly no advantage in underbidding.

All the above cases show that our mechanism is DSIC a posteriori to the click realizations. 
Also, in each of the above cases, note that the utility of an agent $i$, $u_i(\theta_i, b_{-i}, \rho, t) \geq 0$. Therefore, by truthful bidding he never gets a negative utility. This proves that our mechanism is individually rational.
\end{proof}
We next discuss the regret incurred by $\Delta$-UCB. {\color{blue} We note that the regret analysis we provide differs in spirit from the worst case analysis in \cite{Babaioff2014}.  The number of exploration rounds in \cite{Babaioff2014} is required to be $\Omega(T^{2/3})$ since the separation between the best and second best arm is fixed in an adversarial manner in their analysis.  Our analysis does not resort to any adversarial arguments. }

In order to prove our $\Delta$-regret results, we will first need to prove several other lemmas.

\begin{lemma}{Social Welfare UCB index:}
\label{lemma:ucb-index}
For an agent $i$,  we define the social welfare UCB indices for agent $i$ as, 
\begin{align}
\widehat{W}_{i,t}^+  =   \widehat{\mu}_{i,t}\theta_i+ \epsilon_{i,t}\theta_i = \widehat{\mu}_{i,t}\theta_i + \sqrt{2\frac{\theta_i^2 \log T}{N_{i,t}}}
\end{align}
\begin{align}
\widehat{W}_{i,t}^-  =  \widehat{\mu}_{i,t}\theta_i - \epsilon_{i,t}\theta_i =  \widehat{\mu}_{i,t}\theta_i - \sqrt{2\frac{\theta_i^2 \log T}{N_{i,t}}}
\end{align}
Then, $ \forall t \; P\left(\left\lbrace \omega: W_i \notin [\widehat{W}_{i,t}^-(\omega), \widehat{W}_{i,t}^+ (\omega)]) \right\rbrace\right) \leq 2T^{-4} $.
\end{lemma}

\begin{proof}
Let $\widehat{\mu}_{i,t}^+$ and $\widehat{\mu}_{i,t}^-$ denote the UCB and LCB indices for the estimate $\widehat{\mu}_i$.
Then the events $\{\omega: \mu_i \notin [\widehat{\mu}_{i,t}^-(\omega),$ $ \widehat{\mu}_{i,t}^+(\omega) ] \}$ and $\{\omega: W_i \notin [\widehat{W}_{i,t}^-(\omega), \widehat{W}_{i,t}^+(\omega) ] \}$ are identical. So, $P(W_i \notin [\widehat{W}_{i,t}^-, \widehat{W}_{i,t}^+]) = P(\mu_i \notin [\widehat{\mu}_{i,t}^-, \widehat{\mu}_{i,t}^+] )$. An application of Hoeffding bound [\cite{hoeffding1963probability}] gives $P(\mu_i \notin [\widehat{\mu}_{i,t}^-, \widehat{\mu}_{i,t}^+] ) \leq 2\exp(-2 N_{i,t} \epsilon_{i,t}^2) $. As per the mechanism $\epsilon_{i,t} = \sqrt{2\log T/N_{i,t}}$. So, \\ $ P(\mu_i \notin [\widehat{\mu}_{i,t}^-, \widehat{\mu}_{i,t}^+] ) \leq 2\exp(-2 N_{i,t} \times 2 \log T/ N_{i,t}) = 2T^{-4}$.
\end{proof}

\begin{lemma}
\label{lemma:epsilon-delta-relation}
Suppose at time step $t$, $N_{i,t} > \frac{8\theta_{max}^2 \log T}{\Delta^2} \; \forall i \in [K]$. Then $\forall i \in  [K]$,
$2\epsilon_{i,t}\theta_i < \Delta $.
\end{lemma}

\begin{proof}
Given that $N_{i,t} > \frac{8 \theta_{max}^2 \log T}{\Delta^2}$. Therefore,
\begin{align*}
\Delta^2 > \frac{8 \theta_{max}^2 \log T}{N_{i,t}}  
\geq  \frac{8 \theta_i ^2 \log T}{N_{i,t}}  \geq  4 \left[\frac{2 \theta_i ^2 \log T}{N_{i,t}} \right]
\end{align*}
Taking square roots on both sides of the above equation yields
$
\Delta > 2 \epsilon_{i,t} \theta_i
$
thereby proving the lemma.
\end{proof}

\begin{lemma}
\label{lemma:good-event}
Suppose $K  \ll T$. For an agent $i$ and time step $t$, let $B_{i,t}$ be the event $B_{i,t} = \{\omega: W_i \notin  [\widehat{W}_{i,t}^-, \widehat{W}_{i,t}^+] \}$. Define the event $G =\bigcap\limits_t \bigcap\limits_{i \in [K]}  B_{i,t}^c$, where $B_{i,t}^c$ is the complement of $B_{i,t}$.
Then $P(G) \geq 1 - \frac{2}{T^2}$.
\end{lemma}
\begin{proof}
From \Cref{lemma:ucb-index}, the probability of the `bad' event, $P(B_{i,t}) \leq 2T^{-4}$. 
\begin{align*}
P(G)&= P\left(\bigcap\limits_t \bigcap\limits_i B_{i,t}^c\right) = 1- P\left(\left(\bigcap\limits_t \bigcap\limits_i B_{i,t}^c\right)^c\right) \\
&= 1- P\left(\bigcup\limits_t \bigcup\limits_i B_{i,t}\right)  = 1 - \sum_t \sum_{i \in [K]} P(B_{i,t}) \\
& \geq 1- \sum_t \sum_{i \in [K]} 2T^{-4} \geq 1 - \frac{2}{T^2}
\end{align*}
The last statement follows by summing over all rounds and using the fact that $K \ll T$.
\end{proof}

\begin{theorem}
\label{theorem:pull_best_arm}
Suppose at time step $t$, $N_{j,t} > \frac{8\theta_{max}^2 \log T}{\Delta^2} \forall j \in [K]$. Then $\forall i \in  [K]\setminus S_{\Delta} $, 
$\widehat{W}_{i_*, t}^+ > \widehat{W}_{i, t}^+$ with high probability ($= 1 - 2/T^4$).
\end{theorem}
\emph{Proof:}
In \Cref{theorem:single-slot-truthful},  we have shown that $\Delta$-UCB is DSIC. Therefore, all the agents bid their valuations truthfully, $b_i = \theta_i \;\forall i \in [K]$.
Suppose in exploitation round $t$, a sub-optimal arm $i$ is pulled. Therefore, $\widehat{W}_{i, t}^+ \geq \widehat{W}_{i_*, t}^+$. Then one of the following three conditions must have happened.\\
\textbf{Condition 1: $W_i < \widehat{W}_{i,t}^-$}.  This condition implies a drastic overestimate of the sub-optimal arm $i$ so that the true social welfare $W_i$ is even below the LCB index  $\widehat{W}_{i,t}^-$. \Cref{lbl1} shows this case.
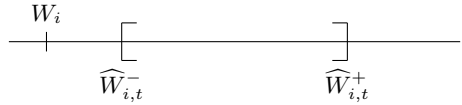
\begin{figure}[h!]
\centering
\begin{tikzpicture}
\coordinate (xbegin) at (0,0);%
\coordinate (p1) at (0.5,0);%
\coordinate (p2) at (1.5,0);%
\coordinate (p3) at (4.5,0);%
\coordinate (xend) at (6,0);%
\coordinate (d) at (0,0.13);%
\coordinate (d1) at (0,0.25);%
\coordinate (dx) at (0.2, 0);%
\draw [-,=>latex] ($ (p1) + (d)$) --($ (p1) - (d)$);%
\draw [-,=>latex] ($ (p2) + (d1)$) --($ (p2) - (d1)$);%
\draw [-,=>latex] ($ (p2) + (d1) $) --($ (p2) + (d1) +(dx) $);%
\draw [-,=>latex] ($ (p2) - (d1) $) --($ (p2) - (d1) +(dx) $);%
\draw [-,=>latex] ($ (p3) + (d1)$) --($ (p3) - (d1)$);%
\draw [-,=>latex] ($ (p3) + (d1) $) --($ (p3) + (d1) -(dx) $);%
\draw [-,=>latex] ($ (p3) - (d1) $) --($ (p3) - (d1) -(dx) $);%
\node [above] at ($ (p1) + (d)$) {$W_{i}$};%
\node [below] at ($ (p2) - (d1)$) {$\widehat{W}_{i,t}^-$};%
\node [below] at ($ (p3) - (d1)$) {$\widehat{W}_{i,t}^+$};%
\draw [-,=>latex] (xbegin)--(xend);%
\end{tikzpicture}
\caption{Condition 1, proof of \Cref{theorem:pull_best_arm}}
\label{lbl1}
\end{figure}\\
\textbf{Condition 2: $ W_{*} > \widehat{W}_{i_*, t}^{+}$}. This implies an underestimate of the optimal arm so that the true social welfare 
$W_{*} $ lies above even the UCB index  $\widehat{W}_{i_*,t}^+$. 
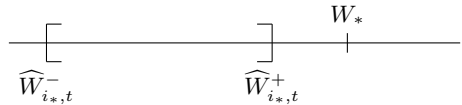
\begin{figure}[h!]
\centering
\begin{tikzpicture}
\coordinate (xbegin) at (0,0);
\coordinate (p2) at (0.5,0);
\coordinate (p3) at (3.5,0);
\coordinate (p1) at (4.5,0);
\coordinate (xend) at (6,0);
\coordinate (d) at (0,0.13);
\coordinate (d1) at (0,0.25);
\coordinate (dx) at (0.2, 0);
\draw [-,=>latex] ($ (p1) + (d)$) --($ (p1) - (d)$);

\draw [-,=>latex] ($ (p2) + (d1)$) --($ (p2) - (d1)$);
\draw [-,=>latex] ($ (p2) + (d1) $) --($ (p2) + (d1) +(dx) $);
\draw [-,=>latex] ($ (p2) - (d1) $) --($ (p2) - (d1) +(dx) $);

\draw [-,=>latex] ($ (p3) + (d1)$) --($ (p3) - (d1)$);
\draw [-,=>latex] ($ (p3) + (d1) $) --($ (p3) + (d1) -(dx) $);
\draw [-,=>latex] ($ (p3) - (d1) $) --($ (p3) - (d1) -(dx) $);

\node [above] at ($ (p1) + (d)$) {$W_{*}$};
\node [below] at ($ (p2) - (d1)$) {$\widehat{W}_{i_*,t}^-$};
\node [below] at ($ (p3) - (d1)$) {$\widehat{W}_{i_*,t}^+$};

\draw [-,=>latex] (xbegin)--(xend);

\end{tikzpicture}
\caption{Condition 2, proof of \Cref{theorem:pull_best_arm}}
\label{lbl2}
\end{figure}\\
\textbf{Condition 3: $W_{*} - W_i < 2\epsilon_{i,t}\theta_i $. } This implies an overlap in the confidence intervals of the optimal and sub-optimal arm. Even though Conditions 1 and 2 are false, still the UCB of sub-optimal arm $i$ is greater than the UCB of the optimal arm $i_*$. 
\begin{figure}[h!]
\centering
\begin{tikzpicture}
\coordinate (xbegin) at (0,0);
\coordinate (p1) at (0.5,0);
\coordinate (p2) at (1.5,0);
\coordinate (p3) at (0.9,0);
\coordinate (p4) at (4.0,0);
\coordinate (p5) at (4.5,0);
\coordinate (p6) at (5.5,0);
\coordinate (xend) at (6,0);

\coordinate (d) at (0,0.13);
\coordinate (d1) at (0,0.25);
\coordinate (dx) at (0.2, 0);
\draw [-,=>latex] (xbegin)--(xend);

\draw [-,=>latex] ($ (p1) + (d1)$) --($ (p1) - (d1)$);
\draw [-,=>latex] ($ (p1) + (d1) $) --($ (p1) + (d1) +(dx) $);
\draw [-,=>latex] ($ (p1) - (d1) $) --($ (p1) - (d1) +(dx) $);

\draw [-,=>latex] ($ (p2) + (d1)$) --($ (p2) - (d1)$);
\draw [-,=>latex] ($ (p2) + (d1) $) --($ (p2) + (d1) +(dx) $);
\draw [-,=>latex] ($ (p2) - (d1) $) --($ (p2) - (d1) +(dx) $);

\draw [-,=>latex] ($ (p3) + (d)$) --($ (p3) - (d)$);

\draw [-,=>latex] ($ (p4) + (d)$) --($ (p4) - (d)$);

\draw [-,=>latex] ($ (p5) + (d1)$) --($ (p5) - (d1)$);
\draw [-,=>latex] ($ (p5) + (d1) $) --($ (p5) + (d1) -(dx) $);
\draw [-,=>latex] ($ (p5) - (d1) $) --($ (p5) - (d1) -(dx) $);

\draw [-,=>latex] ($ (p6) + (d1)$) --($ (p6) - (d1)$);
\draw [-,=>latex] ($ (p6) + (d1) $) --($ (p6) + (d1) -(dx) $);
\draw [-,=>latex] ($ (p6) - (d1) $) --($ (p6) - (d1) -(dx) $);

\node [above] at ($ (p1) + (d)$) {$\widehat{W}_{i,t}^-$};
\node [above] at ($ (p2) + (d)$) {$\widehat{W}_{i_*,t}^-$};
\node [above] at ($ (p5) + (d)$) {$\widehat{W}_{i_*,t}^+$};
\node [above] at ($ (p6) + (d)$) {$\widehat{W}_{i,t}^+$};

\node [below] at ($ (p3) - (d1)$) {$W_{i}$};
\node [below] at ($ (p4) - (d1)$) {$W_{*}$};

\draw [-,=>latex] (xbegin)--(xend);

\end{tikzpicture}
\caption{Condition 3, proof of \Cref{theorem:pull_best_arm}}
\label{lbl3}
\end{figure}

From  \Cref{lbl3}, 
$W_* - W_i \leq \widehat{W}_{i,t}^+ - \widehat{W}_{i,t}^-  \leq  \; 2 \epsilon_{i,t} \theta_i
$

If all the three conditions above were false, then, 
\begin{align*}
\widehat{W}_{i_*,t}^+ &> W_*  >  W_i + 2 \epsilon_{i,t} \theta_i
   > \widehat{W}_{i,t}^- +  2 \epsilon_{i,t} \theta_i= \widehat{W}_{i,t}^+ 
\end{align*}

This implies that $\widehat{W}_{i_*,t}^+ >  \widehat{W}_{i,t}^+$, leading to a contradiction.

As per the statement of the theorem, $N_{i,t} > \frac{8 \theta_{max}^2 \log T}{\Delta^2}$.  Therefore by \Cref{lemma:epsilon-delta-relation}, $2\epsilon_{i,t}\theta_i < \Delta$. For $i \in [K] \setminus S_{\Delta}$, $W_{*} - W_{i} > \Delta > 2\epsilon_{i,t}\theta_i$. So Condition 3 above does not hold true. So if the sub-optimal arm $i$ must have been pulled, only possibilities are for Condition 1 or 2.
\begin{align*}
P(\widehat{W}_{i,t}^+ >& \widehat{W}_{i_*,t}^+) \leq P(\text{Condition 1}) + P(\text{Condition 2}) \\
& \leq \frac{1}{2}P(B_{i,t}) + \frac{1}{2}P(B_{i_*, t}) \leq 2/T^{-4}
\end{align*}
\begin{align*}
P(\widehat{W}_{i_*,t}^+ & >\widehat{W}_{i,t}^+ ) = 1- P(\widehat{W}_{i,t}^+ > \widehat{W}_{i_*,t}^+) \geq 1 - \frac{2}{T^4}
\end{align*}
thereby completing the proof.

We are now ready to state our main result on the incurred regret.

\begin{theorem}
If the $\Delta$-UCB mechanism is executed for a total time horizon of $T$ rounds, it achieves an 
expected $\Delta$-regret of $O(\log T)$.
\label{theorem:Delta-ucb-regret}
\end{theorem}
\begin{proof}
The main idea in the proof is to compute the $\Delta$-regret conditional on two events - $G$ and $G^c$ and then to find a bound for these two conditional expectations.
\begin{align*}
\mathbb{E}&\left[ \Delta\text{-regret} |G \right] = \mathbb{E}\left[ \Delta\text{-regret} | \forall t, \forall i \; W_{i} \in [\widehat{W}_{i,t}^-,\widehat{W}_{i,t}^+ ] \right] \\
&=\mathbb{E}\left[ \sum_{t=1}^T \left(W_* - W_{I_t} \right) \mathbbm{1}\left[ I_t \in [K]\setminus S_\Delta\right] |\forall t, \forall i \; W_{i} \in [\widehat{W}_{i,t}^-,\widehat{W}_{i,t}^+ ]  \right] \\
&=\mathbb{E}\left[ \sum_{t=1}^T \left(W_* - W_{I_t} \right) \mathbbm{1}\left[ I_t \in [K]\setminus S_\Delta\right] | W_{I_t} \in [\widehat{W}_{I_t,t}^-,\widehat{W}_{I_t,t}^+ ]  \right]\\
&\leq \frac{8K \theta_{max}^3 \log T}{\Delta^2} 
\end{align*}
The last step comes from the fact that Conditions 1 and 2 in the proof of \Cref{theorem:pull_best_arm} are eliminated  as we are given that the event $G$ has occurred. After exploration rounds, $N_{i,t}\geq 8K \theta_{max}^2 \log T/\Delta^2$. From \Cref{theorem:pull_best_arm}, no $\Delta$-regret occurs during exploitation since $G$ is true. Therefore the regret is only incurred during the exploration rounds.


We now compute $\mathbb{E}\left[ \Delta\text{-regret} |G^c \right]$.
\begin{align}
\mathbb{E}\left[ \Delta\text{-regret} |G^c \right] \leq T \theta_{max}
\end{align}
But $P(G^c) = 1 - P(G)< \frac{2}{T^2}$ from \Cref{lemma:good-event}. \\
Putting all the steps together,
\begin{align}
\mathbb{E}&\left[ \Delta\text{-regret} \right] = \mathbb{E}\left[ \Delta\text{-regret} |G \right] P(G) + \mathbb{E}\left[ \Delta\text{-regret} |G^c \right] P(G^c) \nonumber \\
 & \leq \frac{8K\theta_{max}^3 \log T}{\Delta^2} * 1 + 
    T\theta_{max}*\frac{2}{T^2} \nonumber\\
  & \leq \frac{8K\theta_{max}^3 \log T}{\Delta^2} + 2
\end{align}
The second term is less than 2 as $\theta_{max} \ll T$.
This completes the proof.
\end{proof}
{
\color{blue}
A consequence of the above theorem is that  even if an adversary chooses an arbitrary small gap between the best and second best arm, there is nothing to worry for the planner - if the gap is less than his tolerance $\Delta$, no loss is incurred as opposed to the otherwise $\Omega(T^{2/3})$ loss in \cite{Babaioff2014}.

\subsection{A Lower Bound for $\Delta$-regret}

We will now discuss a lower bound for the $\Delta$-regret incurred by our approach. In particular, we will provide the lower bound for the  case where $\theta_i = 1$ for all $i$ and is known. The proof will follow along the lines of the lower bound proof in \cite{Bubeck2012a}. The same lower bound will also naturally apply to the case of the general strategic version as well, since we our proposed mechanism $\Delta$-UCB is truthful and achieves a matching upper bound. 

Let $kl(p,q)$ denote the KL divergence between the distributions Bernoulli($p$) and Bernoulli($q$). Then $kl(p,q) = p \log p/q  + (1-p) \log (1-p)/(1-q)$.
\begin{theorem}
\label{thm:lower_bound_non_strategic}
Consider the setting where $\theta_i = 1 \forall i \in [K].$ Suppose an algorithm satisfies $\mathbb{E}[N_{i,t}] = o(t^a)$ for any set of Bernoulli reward distributions and for all arms $i \notin S_{\Delta}$ and $a >0$. Then  
for any set of Bernoulli reward distributions we have,
\begin{equation}
\liminf_{T \rightarrow \infty} \frac{\mathbb{E}[\Delta\text{-regret}]}{\log T } \geq \sum_{i \notin S_{\Delta}} \frac{\Delta_i}{kl(\mu_i, \mu^* + \Delta)} 
\end{equation}
where $\mu^* = \argmax_{j \in [K]} \mu_j$, $\Delta_i = \mu^* - \mu_i$ for all $j \in [K]$.
\end{theorem}

\begin{proof}
We will provide the proof for the case of two agents. The proof for the case $ K > 2$ follows analogously. Assume that $\mu_2 \leq \mu_1 \leq 1$ and $ \mu_1 - \mu_2 > \Delta$. Therefore agent $1$ is optimal and agent 2 does not belong to $S_{\Delta}$.  
For any $\epsilon > 0$, due to the continuity of $kl(\mu_2, x)$, we can find $\mu'_2 \in (\mu_1 + \Delta, 1)$ such that 
\begin{equation}
\label{eq:kl_relation}
kl(\mu_2, \mu'_2) \leq (1+ \epsilon) kl(\mu_2, \mu_1 + \Delta)
\end{equation}
 This configuration then corresponds to an alternate setting where the mean of agent 2 is $\mu'_2$. In this alternate setting, $\mu'_2 - \mu_1 > \Delta$ and agent 2 is the unique optimal.
 For $s \in \{1, \ldots, T \}$, let,
 \begin{equation}
 \tilde{kl}_s = \sum_{t=1}^s \frac{\mu_2  \rho_2^t + (1-\mu_2) (1- \rho_2^t)}{\mu'_2  \rho_2^t + (1-\mu'_2) (1- \rho_2^t)}
 \end{equation}
 It can be verified that $\lim_{t \rightarrow \infty } \mathbb{E}[ \tilde{kl}_t ]/t = kl(\mu_2, \mu'_2)$ (where the expectation is taken over $\rho_2^t$)  and therefore $\tilde{kl}_t$ serves as an un-normalized estimate for $kl(\mu_2, \mu'_2)$.
 
 Let $C_{T}$ denote the following random variable,
 \begin{equation}
 C_T = \mathbbm{1}\{ N_{2,T}  < \frac{(1-\epsilon)\log T}{kl(\mu_2,\mu'_2)} \text{ and } \tilde{kl}_{N_{2,T}} \leq (1 - \epsilon/2) \log T)\}
 \end{equation}
 
 One may verify that $\mathbb{P}_{\mu'_2} (C_T = 1)  = \mathbb{E}_{\mu_2}[C_T \exp(-\tilde{kl}_{N_{2,T}})]$ by applying a change of measure.
 We will now show that $ \mathbb{P}_{\mu_2} (C_T = 1) \rightarrow 0$ as $T \rightarrow \infty$. This is due to the following:
 \begin{align*}
 \mathbb{P}_{\mu'_2} (C_T = 1)  = \mathbb{E}_{\mu_2}[C_T \exp(-\tilde{kl}_{N_{2,T}})] \geq \exp (- (1- \epsilon/2) \log T) \times \mathbb{P}_{\mu_2}(C_T = 1)
 \end{align*}
 Therefore, setting $f_T = \frac{(1-\epsilon)\log T}{kl(\mu_2,\mu'_2)}$, and applying Markov inequality we get,
 \begin{align*}
   \mathbb{P}_{\mu_2}(C_T = 1)& \leq T^{1- \epsilon/2}  \mathbb{P}_{\mu'_2} (C_T = 1) \leq T^{1- \epsilon/2}  \mathbb{P}_{\mu'_2} (N_{2,t} \leq f_T)\\
   &\leq T^{1- \epsilon/2} \frac{\mathbb{E}_{\mu'_2}[T - N_{2,T}]}{T-f_T}   \rightarrow 0
\end{align*}
The last step arises as a consequence of  $T - N_{2,T} = N_{1,T}$ and agent 1 is sub-optimal  for the setting where agent 2 has the mean reward of $\mu'_2$.

We will finally show that $\mathbb{P}_{\mu_2}(N_{2,T} < f_T) \rightarrow 0$.
\begin{align*}
\mathbb{P}_{\mu_2}(C_T = 1)  &\geq \mathbb{P}_{\mu_2}(N_{2,T} < f_T \text{ and } \max_{s \leq f_T} \tilde{kl}_s \leq (1-\epsilon/2)\log T) \\
&= \mathbb{P}_{\mu_2}(N_{2,T} < f_T \text{ and } \frac{kl(\mu_2, \mu'_2)}{(1-\epsilon) \log T} \max_{s \leq f_T} \tilde{kl}_s \leq \frac{kl(\mu_2, \mu'_2)}{(1-\epsilon)}(1-\epsilon/2)) 
\end{align*}
Note that $kl(\mu_2, \mu'_2) > 0$ and $\frac{1-\epsilon/2}{1 - \epsilon} \geq 1$. Therefore by an application of the strong law of  large numbers, we have
\begin{align*}
\lim_{T \rightarrow \infty} \mathbb{P}_{\mu_2}( \frac{kl(\mu_2, \mu'_2)}{(1-\epsilon) \log T} \max_{s \leq f_T} \tilde{kl}_s \leq \frac{kl(\mu_2, \mu'_2)}{(1-\epsilon)}(1-\epsilon/2)) = 1
\end{align*}
Since $\mathbb{P}_{\mu_2}(C_T = 1) \rightarrow 0$, we must have $\mathbb{P}_{\mu_2}(N_{2,T} < f_T) \rightarrow 0$ as well.
Applying Markov inequality again, we get,
\begin{align*}
\mathbb{E}_{\mu_2}[N_{2,T}] \geq \mathbb{P}_{\mu_2}(N_{2,T} \geq f_T) f_T =  \frac{1-\epsilon}{kl(\mu_2,\mu'_2)} \geq  \frac{1-\epsilon}{1+\epsilon} \frac{\log T}{ kl(\mu_2, \mu_1 + \Delta)}
\end{align*}
The last step is obtained by applying \Cref{eq:kl_relation}.  This completes the proof. Note the key difference between our proof and \cite{Bubeck2012a} lies in \Cref{eq:kl_relation}. Our RHS in \Cref{eq:kl_relation} is necessary to ensure that in the alternate scenario agent $1$ is sub-optimal. 
\end{proof}

\begin{remark}
The lower bound for the expected $\Delta$-regret \Cref{thm:lower_bound_non_strategic} is quite similar to the lower bound for the regret of the UCB algorithm in \cite{Bubeck2012a}. The difference is that the KL divergence term in the bound is also a function of the parameter $\Delta$. Intuitively instead of considering the KL divergence between $KL(\mu_i, \mu^*)$, we give an allowance of $\Delta$ for the optimal agent.
\end{remark}

}
\section{Extension to Multi-Slot SSA}
In the previous sections, we assumed that there was a single slot for which the advertisers were competing. We now look at a more general setting where there are $M$ slots to be allocated to the $K$ agents. As before, each advertiser has exactly one ad for display and the CTR for advertisement $i$ is denoted by $\mu_i$. Recall that in the case of single slot auctions, the CTR exactly denoted the probability with which an ad received a click. However in the generalized setting of multi-slot auctions, an additional parameter comes into play while computing the click probability due to which the problem becomes much harder \cite{gatti2015truthful}. 

Each position or slot $m$ is associated with a parameter $\lambda_m$ called `prominence'. $\lambda_m$ denotes the probability with which a user observes an ad at slot $m+1$ given he has observed the ad at slot $m$. In order to understand the need for this parameter, a useful scenario to imagine is the listing of web-pages in Google for a query. There are two phases that one can think of once the listing of pages or results have appeared. \\
\textbf{Phase 1:} This is the phase where a user scans through the pages listed. A page listed higher up in the ranking (say second from the top) has more chances of being observed by a user rather than a page that is far below in the ranking (say fifth from the top). $\lambda_4$, for instance, denotes the probability that a user observes the fifth page, given he has observed the fourth page. Coming back to sponsored ads, we assume that $\lambda_0 =1$, that is, the ad listed in the first slot is surely observed. 
We denote by $\Gamma_m$ the probability that an ad at slot $m$ is observed. $\Gamma_m$ is computed as,
\begin{align}
\Gamma_m =
  \begin{cases}
    1                     &  \quad \text{if } m = 1 \\
   \prod\limits_{s=1}^{m-1} \lambda_s     &  \quad \text{if } 2 \leq m \leq M\\
  0  & \quad \text{if }  m> M\\
  \end{cases}
  \label{eq:gamma_computation}
\end{align}
This modeling assumption for $\Gamma_m$ is known as position dependent cascade model.\\
\textbf{Phase 2:} After having scanned through the list, the user decides to click one or more of the shown ads. In the multi-slot setting \cite{Gatti2012}, it is assumed that multiple ads in a listing may receive clicks. The probability that ad $i$ receives a click when shown at slot $m$ = $\Gamma_m \mu_i$.

We assume that $\lambda_m$, $m=1, \ldots, M$ are known to the planner a-priori. The problem of learning these parameters along with the CTR $\mu$ is much harder in the presence of strategic agents. Therefore, in this section, we work with the assumption that the $\lambda$s and hence $\Gamma$s are known. In \Cref{sec:unknown_gamma}, we give pointers for design of mechanisms where the $\Gamma$s are  unknown.

The above modeling assumptions are as per standard conventions \cite{gatti2015truthful}. In the multi-slot setting, the allocation is given to multiple agents at every time step. We denote by $\mathcal{A}(b, \rho, t)$ $ \subset \{ 1, \ldots, K \}$, the allocation at time $t$ for bids $b$ and click realization $\rho$. The cardinality of the allocated set $|\mathcal{A}(b, \rho, t)| = M$. We also use the notation $\mathcal{A}_i(b, \rho, t) = m$ to denote the allocation to agent $i$ at time $t$ is slot $m$, for the bid profile $b$, click realization $\rho$. If an agent $i$ is not allocated any of the $M$ slots at time $t$, we say $\mathcal{A}_i(b, \rho, t) = 0$.

 We denote by $W_{i,m}$ the social welfare of agent $i$, when he is given slot $m$. $W_{i,m}$ is the expected valuation that agent $i$ receives when he is given slot $m$ and is computed as, 
 \begin{align}
  W_{i,m} = \Gamma_m \mu_i \theta_i
 \end{align}
 For ease of reference, the additional relevant parameters for the multi-slot setting are provided in \Cref{tab:notations-multi-slot}.
 \begin{table}[h!]
\begin{tabular}[c]{|p{0.1\textwidth}|p{0.8\textwidth}|}
\hline
\textbf{Symbol} & \textbf{Description} \\\hline
$M$ & No. of slots \\\hline
$[M]$ & Set of $M$ slots = $\{ 1, \ldots, M\}$ \\\hline
$\lambda_m$ &  Prominence (Probability with which a user observes an ad at slot $m+1$ given he has observed the ad at slot $m$)\\\hline
$\Gamma_m$ & Probability that an ad at slot $m$ is observed\\\hline
$W_{i,m}$ & Social welfare when agent $i$ is allocated slot $m$ \\\hline
$M_{i,t}^{(m)}$ & No. of times agent $i$ has been alloted slot $m$ till time $t$  \\\hline
$N_{i,t}$ & No. of times agent $i$ has been selected till time $t$ over all slots\\\hline
$K^{(m)}$ & Optimal agent for slot $m$\\\hline
$W_{*,m}$ & Social welfare when agent $K^{(m)}$ is allocated slot $m$\\\hline
$S_{\Delta,m}$ & Set of agents whose social welfare is less than $\Delta$ away from $K^{(m)}$. These agents do not contribute to $\Delta$-regret when allocated slot $m$. \\\hline
\end{tabular}
\caption{Additional notations for multi-slot SSA}
\label{tab:notations-multi-slot}
\end{table}

 Having described the multi-slot setting, we now analyze the scenario from the view point of the search engine or central planner. In the ideal scenario, the planner would like to allot the ads exactly to the top $M$ agents with the largest social welfare. This use case has been studied in the literature \cite{Gatti2012} and exploration separated mechanisms with regret of $O(T^{2/3})$ have been proposed.
Various possible allocations are explored for $O(T^{2/3})$ time steps for every agent after which the allocation algorithm is guaranteed to converge to the ideal allocation with high probability. As in the single slot case,  $O(T^{2/3})$ exploration rounds are required to distinguish all the agents perfectly from each other, when there are agents whose social welfare values are arbitrarily close.
 
 However, a much more practical problem of interest is to study and design mechanisms when the search engine is indifferent to a gap in $\Delta$ in social welfare for every slot. We observe that in cases where the agents are well-separated, $O(T^{2/3})$ exploration rounds are not required. In fact, $O(\log T)$ exploration rounds are sufficient to converge to an allocation that is well within the requirements of the search engine.
 
Having explained the problem, we now formalize the notions of separatedness in this setting. 
 Let $K^{(1)},  \ldots, K^{(M)}$ $ \in [K]$ be the best $M$ agents in terms of their single slot social welfare values, that is, $\mu_{K^{(1)}} \theta_{K^{(1)}} > \mu_{K^{(2)}} \theta_{K^{(2)}} > \ldots >\mu_{K^{(M)}} \theta_{K^{(M)}}  $. Let $W_{*,m} = W_{K^{(m)}, m}$. The ideal solution would be to allocate agent $K^{(m)}$ the slot $m$. This allocation would yield the largest social welfare but in the worst case, when the agents' social welfares are separated by a function of $T$, converging to this optimal  allocation would require $O(T^{2/3})$ exploration rounds \cite{Gatti2012}. Instead, for a prescribed value of $\Delta$ fixed by the search engine, define the set,
 \begin{equation}
  S_{\Delta,m } = \left\lbrace i \in [K]:  
 W_{K^{(m)}, m} - W_{i,m} < \Delta \right\rbrace.
 \end{equation} 
 $S_{\Delta,m }$ is the set of all agents whose social welfare is at most $\Delta$ away from the agent $K^{(m)}$ ( who should have ideally been given slot $m$).
 The planner is indifferent to the regret contributed by the agents in $S_{\Delta,m }$, if any of them are allotted slot $m$.
 Hence we define the multi-slot $\Delta$-regret metric as,
 \begin{align*}
 \Delta\text{-regret} = \sum_{t=1}^T \sum_{m=1}^M (W_{*,m} - W_{I_{t,m}, m})\mathbbm{1}\left[ I_{I_t,m} \in [K]\setminus S_{\Delta,m}\right]
 \end{align*}
\noindent The $\Delta$-UCB mechanism for the multi-slot SSA is given in \Cref{alg:delta-ucb-multi-slot}.
 
\begin{algorithm}[h]
\begin{algorithmic}
\Require{ \\
$M:$ No. of slots, $K$: No. of agents, $T$:  Time horizon 
\\$\Delta:$ parameter fixed by the center, $\Gamma_1, \ldots, \Gamma_M$: Slot specific parameters \\ $\theta_{max}:$ Maximum valuation of the agents \\
}
\hrule 
\State Elicit bids $b= (b_1, b_2, \ldots, b_K)$ from all the agents
\State Initialize $\widehat{\mu}_{i,0}  = 0 , N_{i,0} = 0  \; \forall i \in [K]$, 
\State $\gamma= \lceil 8K\theta_{max}^2\log T/\Delta^2 \rceil$
\For{$t= 1, \ldots, \gamma$} \Comment Exploration rounds
\State $\mathcal{A}(b, \rho, t) = \phi $
\For{$m = 1, \ldots, M$}
\State $I_{t,m} = (((t-1)\mod K) + m - 1)\mod K + 1$
\State $N_{I_{t,m}, t} =N_{I_{t,m}, t-1} + 1$
\State $M_{I_{t,m},t}^{(m)} =  M_{I_{t,m},t-1}^{(m)}+ 1$
\State $\mathcal{A}(b, \rho, t) = \mathcal{A}(b, \rho, t)
\cup I_{t,m}$ \Comment Allocate $I_{t,m}$ slot $m$ and observe $\rho_{I_{t,m}} (t)$.
\State $\widehat{\mu}_{I_{t,m},t} = \left(\widehat{\mu}_{I_{t,m},t-1}N_{I_{t,m}, t-1} + \frac{\rho_{I_{t,m}}(t)}{\Gamma_m} \right)/N_{I_{t,m}, t}$
\State $\epsilon_{I_{t,m}, t} = \sqrt{\left(\sum\limits_{m'=1}^M \frac{M_{I_{t,m},t}^{(m')}}{\Gamma_{m'}^2} \right) \frac{2\log T}{N_{I_{t,m}, t}^2}}$
\State $\widehat{\mu}_{I_{t,m},t}^+ = \widehat{\mu}_{I_{t,m},t} + \epsilon_{I_{t,m}, t} $
\State $\widehat{\mu}_{I_t,t}^- = \widehat{\mu}_{I_{t,m},t} - \epsilon_{I_{t,m}, t} $
\EndFor

\State $\widehat{\mu}_{i,t}^+ = \widehat{\mu}_{i,t-1}^+, \widehat{\mu}_{i,t}^- = \widehat{\mu}_{i,t-1}^-  \;\forall i \in [K]\setminus \mathcal{A}(b, \rho, t) $
\State $P_{i}^t (b, \rho) = 0 \;\forall i \in [K]$ \Comment Free rounds
\EndFor

\State  $\widehat{K}^{(1)}, \widehat{K}^{(2)} , \ldots, 
\widehat{K}^{(M)}, \ldots, \widehat{K}^{(K)}$ = sorted list of agents in the decreasing order of $\widehat{\mu}_{i,\gamma}^+ b_i $

\For{$t = \gamma+1, \ldots, T$} 
\Comment Exploitation rounds

\State $\mathcal{A}(b, \rho, t) = \phi$
\For{$m = 1, \ldots, M$}
\State $I_{t,m} = \widehat{K}^{(m)}$
\State $\mathcal{A}(b, \rho, t) = \mathcal{A}(b, \rho, t)
\cup  \widehat{K}^{(m)}$
\State $P_{ \widehat{K}^{(m)}}^t(b, \rho) =  \left(1/\Gamma_m \mu_{\widehat{K}^{(m)},t-1}^+ \right)\sum_{l=m+1}^{M+1} \left(\Gamma_{l-1} - \Gamma_{l}\right) \widehat{\mu}_{K^{(l)},t-1}^+ b_{K^{(l)}} \rho_{\widehat{K}^{(m)}}(t)$
\EndFor

\State $P_{i}^t (b, \rho) = 0 \; \forall i \in [K] \setminus \mathcal{A}(b, \rho, t) $
\State $\widehat{\mu}_{i,t}^+= \widehat{\mu}_{i,\gamma}^+$, $\widehat{\mu}_{i,t}^-= \widehat{\mu}_{i,\gamma}^- \; \forall i \in [K]$  \Comment No more learning
\EndFor
\end{algorithmic}
\caption{$\Delta$-UCB Mechanism for multiple slot SSA}
\label{alg:delta-ucb-multi-slot}
\end{algorithm}

We analyze the regret and truthfulness of \Cref{alg:delta-ucb-multi-slot}. The lemmas and theorems for establishing the results for the multi-slot setting are similar to the single slot setting, however there are subtle differences in proving many of the results. We will highlight them as and when necessary.
\begin{theorem}
In the multi-slot setting $\Delta$-$UCB$ is Dominant Strategy Incentive Compatible (DSIC) and Individually Rational (IR).
\end{theorem}
\begin{proof}
The mechanism is an implementation of the weighted VCG scheme (with the weights for each agent $w_i = \mu_i^+/\mu_i )$ and is hence DSIC and IR.
\end{proof}
 
\begin{lemma}{}
For an agent $i$ and slot $m$, the click through rate UCB indices for agent $i$, 
\begin{align}
\widehat{\mu}_{i,t}^+  = \widehat{\mu}_{i, t}  + \epsilon_{i,t}
  =  \widehat{\mu}_{i, t}  +
  \sqrt{\left(\sum\limits_{m'=1}^M \frac{M_{i,t}^{(m')}}{\Gamma_{m'}^2} \right) \frac{2\log T}{N_{i, t}^2}}
\end{align}
\begin{align}
\widehat{\mu}_{i,t}^- =  \widehat{\mu}_{i,t} - \epsilon_{i,t} = \widehat{\mu}_{i,t}   -  \sqrt{\left(\sum\limits_{m'=1}^M \frac{M_{i,t}^{(m')}}{\Gamma_{m'}^2} \right) \frac{2\log T}{N_{i, t}^2}}
\end{align}
satisfy $P(\mu_{i} \notin [\widehat{\mu}_{i, t}^-, \widehat{\mu}_{i,t}^+])) \leq 2T^{-4} \; \forall t$
\label{lemma:ctr-ucb-index-multi-}
\end{lemma}
\begin{proof}
At every time step, we observe samples $\rho_{I_{t,m}}(t), m = 1, \ldots, M $ corresponding to the clicks of the allocated ads. These samples also encompass slot specific information which must be accounted for in the computation of empirical mean as well as UCB index for $\mu_i$. For an agent $i$,
let the random variable  $C_{i,m}$ denote whether ad $i$ receives a click at slot $m$. Therefore $C_{i,m}$ is a Bernoulli random variable with bias $\Gamma_m \mu_i$.

We obtain a sample $\rho_i(.)$ of $C_{i,m}$ when ad $i$ is allocated slot $m$. However it is the samples from $C_{i,m}/\Gamma_m$ that gives us an unbiased estimator for $\mu_i$. Therefore, the random variable of interest is the Bernoulli random variable,
\begin{align}
D_{i,m} =
  \begin{cases}
    0             &  \quad \text{w.p } 1- \Gamma_m\mu_i \\
   1/\Gamma_m    &  \quad \text{w.p } \Gamma_m\mu_i
  \end{cases}
  \label{eq:X_computation}
\end{align}
$D_{i,m} $ is bounded in $[0, 1/\Gamma_m]$ and $\mathbb{E}[D_{i,m}]$ is $\mu_i$. Also,
\begin{align*}
\log \mathbb{E}\left[\exp(\lambda (D_{i,m} - \mu_i)\right] \leq \frac{\lambda^2}{8 \Gamma_m^2} \text{  (by Hoeffding's Lemma)}
\end{align*} 

Consider the scenario where, for an ad $i$, a single sample click is available from each slot. Let $X_{i,m}$ denote this sample of $C_{i,m}$. Assume $X_{i,m}$ are all independent and $ \widehat{\mu}_i = 1/M\sum_{m=1}^M X_{i,m}/\Gamma_m$. $\mathbb{E}[\widehat{\mu_i}] = \mu_i$. Now,
\begin{align}
P(& \widehat{\mu_i} - \mu_i > \epsilon) =
P\left(\sum_{m=1}^M X_{i,m}/\Gamma_m -M \mu_i > \epsilon M \right) \nonumber \\&= P\left(\exp(\lambda (\sum_{m=1}^M X_{i,m}/\Gamma_m -M \mu_i ))> \exp(\lambda\epsilon M)  \right)  \nonumber \\ 
& \leq \mathbb{E}\left[\exp(\lambda (\sum_{m=1}^M X_{i,m}/\Gamma_m -M \mu_i ))\right]/ \exp(\lambda\epsilon M) \text{  (by Markov inequality)} \nonumber\\
& = \prod\limits_{m=1}^M \mathbb{E}\left[\exp(\lambda (X_{i,m}/\Gamma_m - \mu_i ))\right]/ \exp(\lambda\epsilon M) \text{  (by independence of }X_{i,m}\text{)} \nonumber \\
&= \exp\left(\sum_{m=1}^M \frac{\lambda^2 }{8 \Gamma_m^2} - \lambda M\epsilon\right)
\end{align}
In order to tighten the above bound on the right hand side, one must find appropriate $\lambda$ which minimizes
$\exp(\sum_{m=1}^M \frac{\lambda^2 }{8 \Gamma_m^2} - \lambda M\epsilon)$. Setting $\lambda =  \lambda^*= 4M \epsilon/ \eta$ where $\eta = \sum_{m=1}^M 1/\Gamma_m^2$ achieves the minimum value. 
Therefore,
\begin{align}
P(& \widehat{\mu_i} - \mu_i > \epsilon) \leq \exp(-2M^2 \epsilon^2 /\eta)
\end{align}
In order to obtain a $\delta$ confidence on $P( \widehat{\mu_i} - \mu_i  > \epsilon) $, $\epsilon$ must  be set so that 
$\exp(-2M^2 \epsilon^2 /\eta) = \delta = T^{-4}$. 
Therefore, $\epsilon = \sqrt{\sum\limits_{m=1}^M \left( \frac{1}{\Gamma_m^2} \right) \frac{2 \log T}{M^2}}$.
In the above analysis we assumed that from each slot, one sample was available. When we have a total of $N_{i,t}$ independent samples for ad $i$, with $M_{i,m}^t$ samples for slot $m$ at any time $t$, $\eta = \sum_{m=1}^M M_{i,t}^m/\Gamma_m^2$ and therefore $\epsilon_{i,t} = \sqrt{\left(\sum\limits_{m'=1}^M \frac{M_{i,t}^{(m')}}{\Gamma_{m'}^2} \right) \frac{2\log T}{N_{i, t}^2}}$, completing the proof.
\end{proof}
A noteworthy feature of our estimates is the following. An allocation of an ad $i$ in a slot $m$ yields a sample for the computation of not only $\widehat{W}_{i,m,t}$, but also for $\widehat{W}_{i,m',t}$  for all slots $m' \in \{1, \ldots, M \}$. This is because $\Gamma_m$ is known to the planner a-priori. Therefore note that, the number of allocations that ad $i$ receives till time $t$, $N_{i,t}$ is the sum of the number of allocations that agent $i$ receives irrespective of the slot or inclusive of all the slots.
\begin{lemma}{}
For an agent $i$ and slot $m$, the social welfare UCB indices for agent $i$, 
\begin{align}
\widehat{W}_{i,m,t}^+  = \Gamma_m \widehat{\mu}_{i, t} \theta_i + \epsilon_{i,m,t}
  =  \Gamma_m \widehat{\mu}_{i, t} \theta_i +
  \sqrt{\left(\sum\limits_{m'=1}^M \frac{M_{i,t}^{(m')}}{\Gamma_{m'}^2} \right) \frac{2\theta_i^2 \Gamma_m^2\log T}{N_{i, t}^2}}
\end{align}
\begin{align}
\widehat{W}_{i,m,t}^- = \Gamma_m \widehat{\mu}_{i,  t} \theta_i - \epsilon_{i,m,t} = \Gamma_m \widehat{\mu}_{i,t}  \theta_i  -  \sqrt{\left(\sum\limits_{m'=1}^M \frac{M_{i,t}^{(m')}}{\Gamma_{m'}^2} \right) \frac{2\theta_i^2 \Gamma_m^2\log T}{N_{i, t}^2}}
\end{align}
satisfy $P(W_{i,m} \notin [\widehat{W}_{i, m, t}^-, \widehat{W}_{i,m, t}^+])) \leq 2T^{-4} \; \forall t$
\label{lemma:ucb-index-multi}
\end{lemma}
\begin{proof}
The proof idea is similar to \Cref{lemma:ucb-index}. 
\end{proof}

\begin{lemma}
\label{lemma:epsilon-delta-relation-multi}
Suppose at time step $t$, $N_{j,t} > \frac{8\theta_{max}^2 \log T}{\Delta^2} \; \forall j \in [K]$. Then $\forall i \in  [K] $ and $\forall m \in [M]$,
$
2\epsilon_{i,m, t}  < \Delta.
$
\end{lemma}
\begin{proof}
The proof is similar to \Cref{lemma:epsilon-delta-relation}.
\end{proof}

\begin{lemma}
\label{lemma:good-event-multi}
For an agent $i$, slot $m$ and time $t$, let $B_{i,m,t}$ be the event $B_{i,m,t} = \{\omega: W_{i,m} \notin  [\widehat{W}_{i,m,t}^-(\omega), \widehat{W}_{i,m,t}^+(\omega)] \}$. Define the event $G =\bigcap\limits_t \bigcap\limits_i \bigcap\limits_m  B_{i,m,t}^c$.
Then $P(G) \geq 1 - \frac{2}{T^2}$.
\end{lemma}
\emph{Proof:} The proof has some subtle differences from \Cref{lemma:good-event} because in the multi-slot extension, the events $B_{i,m,t}$ are not independent across the slots.\\
\emph{Observation:} If an element $\omega$ from the set of outcomes is such that $ \omega \in B_{i,m,t}$, then $\omega \in B_{i,m',t} \; \forall m' \in [M]$. This is because, for any two slots $m$ and $m'$,
\begin{align*}
 W_{i,m} \notin  [\widehat{W}_{i,m,t}^-, \widehat{W}_{i,m,t}^+] 
 &\iff \mu_{i} \notin  [\widehat{\mu}_{i,t}^-, \widehat{\mu}_{i,t}^+]\\
  & \iff W_{i,m'} \notin  [\widehat{W}_{i,m',t}^-, \widehat{W}_{i,m',t}^+] 
\end{align*}
Therefore $P(\bigcup_m B_{i,m,t}) = P(B_{i,1,t})$.
From \Cref{lemma:ucb-index-multi},\\ $P(\bigcup_m B_{i,m,t}) = P(B_{i,1,t}) \leq 2T^{-4}$. Hence, 
\begin{align*}
P(G) &= 1 - P\left(\bigcup\limits_t \bigcup\limits_i \bigcup\limits_m B_{i,m,t}\right) = 1 - P\left(\bigcup\limits_t \bigcup\limits_i B_{i,1,t}\right) \\
&\geq 1 - \frac{2}{T^2} \text{ (as in  \Cref{lemma:good-event})}.
\end{align*}

\begin{theorem}
\label{theorem:pull_best_arm-multi}
Suppose at time $t$, $N_{j,t}>8\theta_{max}^2 \log T/\Delta^2 \;  \forall j \in [K]$. Then $ \forall m \in [M], \forall i \in  [K]\setminus S_{\Delta, m}$, 
$\widehat{W}_{K^{(m)}, m, t}^+ > \widehat{W}_{i, m, t}^+$ with high probability ($= 1 - 2/T^4$).
\end{theorem}
\emph{Proof:} Suppose at time $t$ where $N_{j,t} > 8 \theta_{max}^2 \log T/\Delta^2 \;\forall j \in [K], $ there exists some $ m \in [M] $ such that  $\widehat{W}_{K^{(m)}, m, t}^+ < \widehat{W}_{i, m, t}^+$. (Note that this statement does not arise from any assumptions on the allocation, for instance, that agent $i$ is given slot $m$. This is the major difference from \Cref{theorem:pull_best_arm}). But the relation between the true social welfare values of these agents is  $W_{K^{(m)}, m} >  W_{i, m}$. Then one of the following three conditions must have occurred, like in proof of \Cref{theorem:pull_best_arm}. \\
\textbf{Condition 1: $W_{i,m} < \widehat{W}_{i,m,t}^-$}.  This condition implies a drastic overestimate of the sub-optimal arm $i$ so that the true mean social welfare $W_{i,m}$ is even below the LCB index  $\widehat{W}_{i,m,t}^-$. The figure below captures this condition.
\begin{figure}[h!]
\centering
\begin{tikzpicture}
\coordinate (xbegin) at (0,0);
\coordinate (p1) at (0.5,0);
\coordinate (p2) at (1.5,0);
\coordinate (p3) at (4.5,0);
\coordinate (xend) at (6,0);
\coordinate (d) at (0,0.13);
\coordinate (d1) at (0,0.25);
\coordinate (dx) at (0.2, 0);
\draw [-,=>latex] ($ (p1) + (d)$) --($ (p1) - (d)$);

\draw [-,=>latex] ($ (p2) + (d1)$) --($ (p2) - (d1)$);
\draw [-,=>latex] ($ (p2) + (d1) $) --($ (p2) + (d1) +(dx) $);
\draw [-,=>latex] ($ (p2) - (d1) $) --($ (p2) - (d1) +(dx) $);

\draw [-,=>latex] ($ (p3) + (d1)$) --($ (p3) - (d1)$);
\draw [-,=>latex] ($ (p3) + (d1) $) --($ (p3) + (d1) -(dx) $);
\draw [-,=>latex] ($ (p3) - (d1) $) --($ (p3) - (d1) -(dx) $);

\node [above] at ($ (p1) + (d)$) {$W_{i,m}$};
\node [below] at ($ (p2) - (d1)$) {$\widehat{W}_{i,m,t}^-$};
\node [below] at ($ (p3) - (d1)$) {$\widehat{W}_{i,m,t}^+$};

\draw [-,=>latex] (xbegin)--(xend);

\end{tikzpicture}
\caption{Condition 1, Proof of \Cref{theorem:pull_best_arm-multi}}
\label{lbl4}
\end{figure}\\
\textbf{Condition 2: $ W_{K^{(m)},m} > \widehat{W}_{K^{(m)},m, t}^{+}$}. This implies an underestimate of the optimal arm so that the true mean social welfare 
$W_{K^{(m)},m} $ lies above even the UCB index $\widehat{W}_{K^{(m)},m,t}^+$. See \Cref{lbl5} below.\\
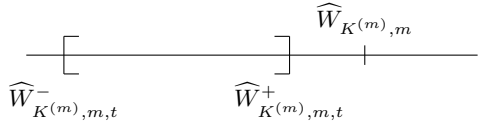
\begin{figure}[h!]
\centering
\begin{tikzpicture}
\coordinate (xbegin) at (0,0);
\coordinate (p2) at (0.5,0);
\coordinate (p3) at (3.5,0);
\coordinate (p1) at (4.5,0);
\coordinate (xend) at (6,0);
\coordinate (d) at (0,0.13);
\coordinate (d1) at (0,0.25);
\coordinate (dx) at (0.2, 0);
\draw [-,=>latex] ($ (p1) + (d)$) --($ (p1) - (d)$);

\draw [-,=>latex] ($ (p2) + (d1)$) --($ (p2) - (d1)$);
\draw [-,=>latex] ($ (p2) + (d1) $) --($ (p2) + (d1) +(dx) $);
\draw [-,=>latex] ($ (p2) - (d1) $) --($ (p2) - (d1) +(dx) $);

\draw [-,=>latex] ($ (p3) + (d1)$) --($ (p3) - (d1)$);
\draw [-,=>latex] ($ (p3) + (d1) $) --($ (p3) + (d1) -(dx) $);
\draw [-,=>latex] ($ (p3) - (d1) $) --($ (p3) - (d1) -(dx) $);

\node [above] at ($ (p1) + (d)$) {$\widehat{W}_{K^{(m)},m}$};
\node [below] at ($ (p2) - (d1)$) {$\widehat{W}_{K^{(m)},m,t}^-$};
\node [below] at ($ (p3) - (d1)$) {$\widehat{W}_{K^{(m)},m,t}^+$};

\draw [-,=>latex] (xbegin)--(xend);
\end{tikzpicture}
\caption{Condition 2, Proof of \Cref{theorem:pull_best_arm-multi}}
\label{lbl5}
\end{figure}\\
\textbf{Condition 3:} $W_{K^{(m)},m} - W_{i,m} < 2\epsilon_{i,m,t}$. This implies an overlap in the confidence intervals of the optimal and sub-optimal arm. Even if, Conditions 1 and 2 are false, still the UCB of sub-optimal arm $i$ is greater than the UCB of the optimal arm $i_*$.
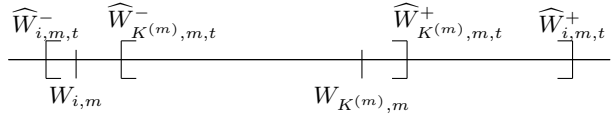
\begin{figure}[h!]
\centering
\begin{tikzpicture}
\coordinate (xbegin) at (0,0);
\coordinate (p1) at (0.5,0);
\coordinate (p2) at (1.5,0);
\coordinate (p3) at (0.9,0);
\coordinate (p4) at (4.7,0);
\coordinate (p5) at (5.3,0);
\coordinate (p6) at (7.5,0);
\coordinate (xend) at (8,0);

\coordinate (d) at (0,0.13);
\coordinate (d1) at (0,0.25);
\coordinate (dx) at (0.2, 0);
\draw [-,=>latex] (xbegin)--(xend);

\draw [-,=>latex] ($ (p1) + (d1)$) --($ (p1) - (d1)$);
\draw [-,=>latex] ($ (p1) + (d1) $) --($ (p1) + (d1) +(dx) $);
\draw [-,=>latex] ($ (p1) - (d1) $) --($ (p1) - (d1) +(dx) $);

\draw [-,=>latex] ($ (p2) + (d1)$) --($ (p2) - (d1)$);
\draw [-,=>latex] ($ (p2) + (d1) $) --($ (p2) + (d1) +(dx) $);
\draw [-,=>latex] ($ (p2) - (d1) $) --($ (p2) - (d1) +(dx) $);

\draw [-,=>latex] ($ (p3) + (d)$) --($ (p3) - (d1)$);

\draw [-,=>latex] ($ (p4) + (d)$) --($ (p4) - (d1)$);

\draw [-,=>latex] ($ (p5) + (d1)$) --($ (p5) - (d1)$);
\draw [-,=>latex] ($ (p5) + (d1) $) --($ (p5) + (d1) -(dx) $);
\draw [-,=>latex] ($ (p5) - (d1) $) --($ (p5) - (d1) -(dx) $);

\draw [-,=>latex] ($ (p6) + (d1)$) --($ (p6) - (d1)$);
\draw [-,=>latex] ($ (p6) + (d1) $) --($ (p6) + (d1) -(dx) $);
\draw [-,=>latex] ($ (p6) - (d1) $) --($ (p6) - (d1) -(dx) $);

\node [above] at ($ (p1) + (d)$) {$\widehat{W}_{i,m,t}^-$};
\node [above right] at ($ (p2) + (d) + (-0.3,0) $)
{$\widehat{W}_{K^{(m)},m,t}^-$};
\node [above right] at ($ (p5) + (d) + (-0.3,0) $)
{$\widehat{W}_{K^{(m)},m,t}^+$};
\node [above] at ($ (p6) + (d)$) {$\widehat{W}_{i,m,t}^+$};

\node [below] at ($ (p3) - (d1)$) {$W_{i,m}$};
\node [below] at ($ (p4) - (d1)$) {$W_{K^{(m)},m}$};

\draw [-,=>latex] (xbegin)--(xend);
\end{tikzpicture}
\caption{Condition 3, Proof of \Cref{theorem:pull_best_arm-multi}}
\label{lbl6}
\end{figure}
From the figure,
$
W_{K^{(m)},m} - W_{i,m} \leq
\widehat{W}_{i,m,t}^+ - \widehat{W}_{i,m,t}^-  \leq  \; 2 \epsilon_{i, m, t} 
$.
If all the three conditions above were false, then, 
\begin{align*}
\widehat{W}_{K^{(m)},m,t}^+ &> W_{K^{(m)},m}  >  W_{i,m} + 2 \epsilon_{i,t}
   > \widehat{W}_{i,m,t}^- +  2 \epsilon_{i,t} \\
   & = \widehat{W}_{i,m, t}^+ \;\;\;\text{ ( A contradiction!)}
\end{align*}
As per the statement of the theorem, $N_{i,t} > 8 \theta_{max}^2 \log T/\Delta^2$.  Therefore by \Cref{lemma:epsilon-delta-relation-multi}, $2\epsilon_{i,m,t} < \Delta$. For agent $i \in [K]\setminus S_{\Delta,m}$, $W_{K^{(m)},m} - W_{i, m} > \Delta > 2\epsilon_{i,m,t}$. Therefore, Condition 3 above does not hold true. So,
\begin{align*}
P(\widehat{W}_{i,m,t}^+ >& \widehat{W}_{K^{(m)},m,t}^+) \leq P(\text{Condition 1}) + P(\text{Condition 2}) \\
& \leq 0.5 P(B_{i,m,t}) +  0.5P(B_{K^{(m)},m,t}) \leq 2/T^{-4}
\end{align*}
\begin{align*}
P(\widehat{W}_{K^{(m)},m, t}^+ >\widehat{W}_{i,m, t}^+ ) &= 1- P(\widehat{W}_{i,m, t}^+ > \widehat{W}_{K^{(m)},m,t}^+) \geq 1 - \frac{2}{T^4}
\end{align*}
\begin{theorem}
If the $\Delta$-UCB mechanism is executed in the multiple slot scenario for a total time horizon of $T$ rounds, it achieves an expected $\Delta$-regret of $O(\log T)$.
\label{theorem:Delta-ucb-regret-multi}
\end{theorem}

\begin{proof}
The proof idea has some subtle differences from the proof of \Cref{theorem:Delta-ucb-regret}. 
As before, we first compute the expected $\Delta$-regret conditional on $G$.
For the exploration rounds, the mechanism obtains a regret of $\xi=\frac{8MK \theta_{max}^3 \log T}{\Delta^2}$. 
\begin{align*}
\mathbb{E}&\left[ \Delta\text{-regret}|G \right] \leq  
 \xi +  \sum_{t=\gamma+1}^T \sum_{m=1}^M (W_{K^{(m),m}} - W_{(I_{t,m}),m})
\mathbbm{1}\left[I_{t,m} \in K \setminus S_{\Delta,m}|G\right] 
\end{align*}
We will now show that the second term above evaluates to zero. 
For any $m$, the cardinality of $S_{\Delta,m}$ is at least $m$. This is because for all $K^{(j)}$ above $m$ in the ranking of agents ($j<m$), $W_{K^{(m)}, m} - W_{K^{(j)},m} < 0 < \Delta$ as $W_{K^{(j)},m}> W_{K^{(m)}, m}$. Therefore there are at least $m-1$ agents in $S_{\Delta,m}$. Also $K^{(m)} \in  S_{\Delta,m}$ as $W_{K^{(m)}, m} - W_{K^{(m)},m}=0 < \Delta$. Therefore $ \forall j \in \{1, \ldots, m\}, K^{(j)} \in S_{\Delta,m}$. While allocating slot $m$, at least one of the agents in $S_{\Delta,m}$ must be free. This is by the pigeonhole principle.
  Now if the allocated agent for slot $m$, $I_{t,m} \in [K] \setminus S_{\Delta,m}$, one of the following two cases occur.\\
\textbf{Case 1:} The ideal agents $K^{(1)}, \ldots, K^{(m-1)}$ for all the previous slots ${1, \ldots, m-1}$ have already been allocated before the allocation of slot $m$. This means that $K^{(m)}$ has not been allocated yet. Also, 
$\widehat{W}_{(I_{t,m}),m, \gamma}^+ > \widehat{W}_{K^{(m)},m, \gamma}^+ $.
Since $G$ is true and $t > \gamma$, the above event cannot occur (by \Cref{theorem:pull_best_arm-multi}).\\
\textbf{Case 2:} The agent $K^{(m)}$ has already been allocated to some other slot before the allocation of slot $m$ has begun. Therefore there is some agent $K^{(j)}, j<m$ with a larger social welfare value, who has still not been allocated. That is, $W_{K^{(j)},m} > W_{K^{(m)},m} > W_{(I_{t,m}),m}$. Given that $I_{t,m} \notin S_{\Delta,m}$.
Therefore we can deduce that $I_{t,m} \notin S_{\Delta,j}$.
This is because,
\begin{align}
&W_{K^{(m)},m}  - W_{(I_{t,m}),m} \geq   \Delta 
\nonumber \\& \implies W_{K^{(j)},m} - W_{(I_{t,m}),m} \geq   \Delta
 \nonumber\\& \implies \mu_{K^{(j)}} \theta_{K^{(j)}} -  \mu_{I_{t,m}} \theta_{I_{t,m}}  \geq   \Delta/\Gamma_m
\nonumber\\&  \implies \Gamma_j (\mu_{K^{(j)}} \theta_{K^{(j)}} -  \mu_{I_{t,m}} \theta_{I_{t,m}})  \geq  \Gamma_j \Delta/\Gamma_m \nonumber \\& \implies W_{K^{(j)},j} - W_{(I_{t,m}),j} \geq   \Delta  
\end{align}
The last line in the above implications is true as $\Gamma_j > \Gamma_m $. But $\widehat{W}_{K^{(j)},m, \gamma}^+ < \widehat{W}_{(I_{t,m}),m, \gamma}^+$. Then the inequality
$\widehat{W}_{K^{(j)},j, \gamma}^+ < \widehat{W}_{(I_{t,m}),j, \gamma}^+$ is also true due to the way the slot specific UCB indices are computed. 
From \Cref{theorem:pull_best_arm-multi} for slot $j$, we find that $\widehat{W}_{K^{(j)},j, \gamma}^+ > \widehat{W}_{(I_{t,m}),j, \gamma}^+$. Again this cannot happen as $G$ is true and $t > \gamma$.
Therefore we get that $\mathbb{E}\left[ \Delta\text{-regret}|G \right] \leq  
 \xi$. \\
Also, $P(G^c) = 1 - P(G)< \frac{2}{T^2}$ from \Cref{lemma:good-event-multi}. \\
Putting all the steps together,
\begin{align}
\mathbb{E}\left[ \Delta\text{-regret}\right] &= \mathbb{E}\left[ \Delta\text{-regret} |G \right] P(G) + \mathbb{E}\left[ \Delta\text{-regret} |G^c \right] P(G^c) \nonumber \\
 & \leq \frac{8KM\theta_{max}^3 \log T}{\Delta^2} * 1 + 
    TM\theta_{max}*\frac{2}{T^2}  \nonumber\\
  & \leq \frac{8KM\theta_{max}^3 \log T}{\Delta^2} + 2\theta_{max}
\end{align}
The simplification in the second line is because $\mathbb{E}\left[\Delta\text{-regret} |G^c \right]$ $ \leq TM \theta_{max}$. In the last line we use the fact that  $M \ll T$.
This completes the proof.
\end{proof}
\section{Extensions to Other Variants of Multi-slot SSA}
In this section, we look at other variants in the multi-slot SSA setting and discuss how our mechanism can be adapted to such settings.
\subsection{Position and Ad Dependent Cascade Model}
We have explained our algorithm and performed the analysis for the position dependent cascade model for SSA where the $\Gamma_m$ function is characterized by \Cref{eq:gamma_computation} and is known to the planner a-priori.
A more general model would be one where the function $\Gamma_m$ may also depend on the ad displayed at position $m$. Our model can also be used in such scenarios and the same analysis will hold.
\subsection{Handling the Case of Unknown $\Gamma_m$}
\label{sec:unknown_gamma}
We have assumed that the functions $\Gamma_m$s are known to the planner a-priori. Now suppose that the $\Gamma_m$s are required to be learnt. The same allocation scheme as in \Cref{alg:delta-ucb-multi-slot} may be used. However the computation of the proposed payment scheme in \Cref{alg:delta-ucb-multi-slot}  is not feasible as the payments use $\Gamma_m$s, which are unknown. 

In order to handle such a scenario, we must obtain estimates for $\Gamma$ first. It is known that, the parameter for the first slot, $\Gamma_1 = 1$. 
Only $\Gamma_2, \ldots, \Gamma_M$ need to be estimated. 
We will first describe a mechanism which relies on  an arbitrary  learning algorithm  to provide estimates $ \widehat{\Gamma}_2, \ldots,\widehat{\Gamma}_M$. Thereafter we will remark on the possible learning schemes. 
\begin{proposition}
\label{prop-estimates}
Suppose we have a learning scheme that gives us estimates
$\widehat{\Gamma}_2, \ldots,\widehat{\Gamma}_M$ such that, 
$\widehat{\Gamma}_2 \geq \widehat{\Gamma}_3 \geq \ldots \geq \widehat{\Gamma}_M$ and $0 \leq \widehat{\Gamma}_m \leq 1 \text{ for } m = 2, \ldots, M$. Let $\widehat{\Gamma}_1 = 1$.
\end{proposition}
We propose a weighted VCG mechanism [\cite{NISAN2007}] which is known to be DSIC truthful and is also IR. Suppose the private valuation of agent $i$ for a click is $\theta_i$. Let $x \in \{0,1\}^{K\times M}$ be an outcome of the allocation. $x_{im} = 1$ if ad $i$ is alloted slot $m$ and zero otherwise. The valuation function of agent $i$ in this case is,
\begin{align}
v_i(x, \theta_i) = \sum_{m=1}^M \Gamma_m \mu_i \theta_i x_{im}
\end{align}
Define a weight vector $w_i \in \mathcal{R}^M$ for every agent $i$. $w_i$ has  weights corresponding to agent $i$ and slot $m$ such that, $w_{i,m} = \frac{\widehat{\mu}_i^+ \widehat{\Gamma}_m}{\mu_i \Gamma_m}$. $\widehat{\mu}_i^+$ is the UCB index corresponding to the CTR of ad $i$, computed after the fixed number of exploration rounds as in \Cref{alg:delta-ucb-multi-slot}. However, in this scenario, the UCB index is constructed using samples of the clicks from allocation in the first slot alone. 
\begin{figure}
\begin{tcolorbox}

First, $\gamma = 8K\theta_{max}^2\log T/\Delta^2$ exploration rounds are performed free for all agents as in \Cref{alg:delta-ucb-multi-slot}. At every time step $t$, UCB indices for every ad $i$, ($\widehat{\mu}_{i,t}^{-} \text{ and } \widehat{\mu}_{i,t}^{+}$) are computed using the update in \Cref{alg:delta-ucb-multi-slot}, but using only samples from the allocation of ad $i$ to slot $1$.

Thereafter, in every round $t$, our weighted VCG mechanism uses the allocation,
\begin{align*}
A^*(b_i, b_{-i}) = \argmax\limits_x \sum_{i=1}^K \sum_{m=1}^M  \Gamma_m \mu_i b_i x_{im} w_{i,m} 
\end{align*} 

The payment for an agent $i$ allocated slot $m'$ where $1 \leq m' \leq M$ is,
\begin{align*}
P_i^t(b,\rho) = \frac{\rho_i(t)}{\widehat{\mu}_{i,t}^+ \widehat{\Gamma}_{m'}} \sum_{j \neq i} \sum_{m=m'+1}^{M+1} \widehat{\mu}_{j,t}^+ b_j x_{jm} (\widehat{\Gamma}_{m-1} - \widehat{\Gamma}_m )
\end{align*}
where $\widehat{\Gamma}_{M+1} = 0$.
\end{tcolorbox}
\caption{$\Delta$-UCB Mechanism for the Position Dependent Cascade Model using Estimates for $\Gamma_m$s}
\label{multi-slot-unknown-lambda-mech}
\end{figure}
Our weighted VCG mechanism is described in \Cref{multi-slot-unknown-lambda-mech}. The mechanism uses the allocation,
\begin{align*}
A^*(b_i, b_{-i}) = \argmax\limits_x \sum_{i=1}^K \sum_{m=1}^M  \Gamma_m \mu_i b_i x_{im} w_{i,m} 
\end{align*} 
But note that this allocation rule boils down to the same  allocation used in \Cref{alg:delta-ucb-multi-slot}. This is due to the fact that the estimates $\widehat{\Gamma}_m$ monotonically decrease with $m$. 
The procedure for obtaining the allocation $A^*(b_i, b_{-i})$ is the following. We sort the agents based on $\widehat{\mu}_i^+ b_i$ and allocate the slots to the best $M$ agents.
Therefore, the allocation rule is independent of the $\Gamma$s and is equivalent to,
\begin{align*}
A^*(b_i, b_{-i}) = \argmax\limits_x \sum_{i=1}^K \sum_{m=1}^M  \widehat{\mu}_i^+ b_i x_{im} 
\end{align*} 
The expected payment to be made by agent $i$ when allocated a slot $m'$ is,
\begin{align*}
\mathbb{E}[P_i^t(b, \rho)] = \frac{\mu_i \Gamma_{m'}}{\widehat{\mu}_{i,t}^+ \widehat{\Gamma}_{m'}} \sum_{j \neq i} \sum_{m=m'+1}^{M+1} \widehat{\mu}_{j,t}^+ b_j x_{jm} (\widehat{\Gamma}_{m-1} - \widehat{\Gamma}_m )
\end{align*}
The above is the externality based payment prescribed by weighted VCG.
However since we adopt the pay per click scheme,
\begin{align*}
P_i^t(b,\rho)] = \frac{\rho_i(t)}{\widehat{\mu}_{i,t}^+ \widehat{\Gamma}_{m'}} \sum_{j \neq i} \sum_{m=m'+1}^{M+1} \widehat{\mu}_{j,t}^+ b_j x_{jm} (\widehat{\Gamma}_{m-1} - \widehat{\Gamma}_m )
\end{align*}
Therefore, the  computation of the payments is also feasible now. The above mentioned weighted VCG scheme is DSIC truthful and IR. The proof follows from the standard weighted VCG scheme where the weights are as defined as above. We now remark on the $\Delta$-regret of the mechanism. 
\subsubsection{Remarks on Learning  $\widehat{\Gamma}_m$ and  Computation of $\Delta$-regret}
In the above mechanism we have assumed, that the estimates $\widehat{\Gamma}_{m}$ satisfy \Cref{prop-estimates}. The allocation scheme described above ultimately does not rely on these estimates, although the weights $w_{i,m}$ use it. The mechanism therefore uses the estimates only in the payment rule. 
We now make an important observation here.\\
\textit{Observation:} When any set of estimates $\{\widehat{\Gamma}_{m}\}$, $m=1, \ldots, M$  satisfying \Cref{prop-estimates} is used in the mechanism above,  the mechanism is DSIC truthful, IR and suffers only logarithmic $\Delta$-regret.

The reason is that the mechanism is an instance of weighted VCG mechanism and therefore is DSIC truthful and IR, with any estimate for the $\Gamma_m$s. As far as the $\Delta$-regret in social welfare is concerned, the allocation rule determines it. The allocation rule used turns out to be identical to the allocation rule used where $\Gamma_m$ is known and is independent of the estimates. Note that it is now possible to minimise regret in payments by choosing estimates  $\widehat{\Gamma}_m$ that maximise the payments and also satisfy the constraints in \Cref{prop-estimates}. This will lead to a constrained optimization problem which can be solved.  However the current work focuses on minimizing $\Delta$-regret in social welfare and therefore the problem of minimising regret in payments is still open.
\section{Conclusion}
We have studied the more practical use case in MAB mechanisms where a planner has the option to specify a tolerance level $\Delta$ for sub-optimal arms. 
All the papers in the literature on MAB mechanisms propose schemes to target the worst case scenario where the arms are arbitrarily close. Therefore they prescribe investing a huge number of exploration rounds ($\Omega(T^{2/3})$) to perfectly distinguish the arms. However, the planner may not want to perfectly distinguish arms that are arbitrarily close. Many a time, the planner may instead be willing to allocate arms that are at most $\Delta$ away from the best arm. The state of the art does not permit this flexibility to the planner. Towards providing such a flexibility to the planner, we have, for the first time, introduced a new notion of regret called $\Delta$-regret. When arms that are less than $\Delta$ away from the best arm are selected, the $\Delta$-regret incurred is zero. Only arms more than $\Delta$ away from the best arm contribute to the $\Delta$-regret. 

From the above perspective, we have revisited the application of MAB mechanisms in sponsored search auctions. First we analysed the single slot SSA setting and proposed a deterministic, exploration separated MAB mechanism called $\Delta$-UCB. We showed that $\Delta$-UCB is DSIC truthful, IR and achieves a  $\Delta$-regret of $O(\log T)$. Next we studied the more challenging setting of multi-slot SSA. In particular, we adopted the cascade model and adapted $\Delta$-UCB to this setting, first with the assumption that the prominence parameters are known. Here too, we have shown that the mechanism  is DSIC truthful, IR and achieves a  $\Delta$-regret of $O(\log T)$. We finally adapt the mechanism to the general multi-slot SSA setting where neither the CTRs nor the prominences are known. Here too our deterministic, exploration separated mechanism is DSIC truthful, IR and suffers a $\Delta$-regret of $O(\log T)$. The other mechanisms in literature for this setting are not able to obtain all these desirable properties that our mechanism achieves. They either compromise on the truthfulness, satisfying a weaker notion (truthfulness in expectation) or are forced to resort to randomness in the mechanism. 
 
 Our results are generic and apply equally well to several other applications where MAB mechanisms have been used.



\bibliographystyle{plain} 
\bibliography{ucb_interval}

\end{document}